\documentclass[onecolumn,10pt]{IEEEtran}
\pdfoutput=1 % for arxiv
\usepackage{ifluatex}
\ifluatex
\usepackage{lualatex-math}
\newcommand*{\Shah}{\textup{ш}}
\else
\usepackage[utf8]{inputenc}
\usepackage[OT2,OT1]{fontenc}
\DeclareSymbolFont{cyrletters}{OT2}{wncyr}{m}{n}
\DeclareMathSymbol{\Shah}{\mathalpha}{cyrletters}{"78}   % for Shah
\usepackage[protrusion=true,expansion, kerning, spacing]{microtype}
\fi
\usepackage{amsfonts}
\usepackage[cmex10]{amsmath}
\usepackage{amssymb, amsthm}
\usepackage{mathtools}
\usepackage{fixltx2e}
\usepackage{paralist, ellipsis, cancel, units, csquotes, siunitx, suffix, xparse}
\usepackage{subfig, float}
\usepackage[section]{placeins}
\newcommand{\coloneqq}{\triangleq}
\newcommand*{\mytitle}{Estimation of Overspread Scattering Functions}
\newcommand*{\myauthor}{G\"{o}tz Pfander, Pavel Zheltov}
\usepackage{hyperref}
\hypersetup{
colorlinks,
unicode=true, % instead of \texorpdfstring
hyperindex=true,
citecolor=blue,
linkcolor=blue,
linkcolor=blue,
urlcolor=blue,
pdfkeywords = {radar, channel estimation, spreading function, scattering function, time-varying operators},
pdfauthor = {\myauthor},
pdftitle = {\mytitle},
pdfpagelayout = OneColumn,
pdfstartview = FitH
}
\usepackage{tikz, pgfplots}
\graphicspath{ {./images/} }
\pgfplotsset{compat=newest}
\usetikzlibrary{decorations.pathreplacing, decorations.pathmorphing, scopes,  shapes.misc, decorations.shapes, decorations.markings, calc, plotmarks}
\usepgfplotslibrary{colormaps}
\colorlet{mycolor}{violet!40}
\pgfplotsset{every linear axis/.append style={
    colorbar,
    colormap/jet,
    point meta min=0,
    xtick={0, 0.333, 0.666, 1, 1.333, 1.6666, 2},%
    xticklabels={0, $\frac{1}{3}$, $\frac23$,1, $1\frac13$, $1\frac23$, 2},
    enlargelimits=false}}
\pgfplotsset{/pgfplots/filter discard warning=false}
\pgfplotsset{every colorbar/.append style={scaled ticks=false, yticklabel style={/pgf/number format/.cd, fixed}}}
\newlength\figureheight
\newlength\figurewidth
\newcommand{\mywidth}{\figurewidth}
\theoremstyle{plain}
\newtheorem{theorem}{Theorem}[section]
\usepackage{aliascnt}
\newaliascnt{lemma}{theorem}

\aliascntresetthe{lemma}

\newaliascnt{proposition}{theorem}

\aliascntresetthe{proposition}

\newaliascnt{corollary}{theorem}
\newtheorem{corollary}[corollary]{Corollary}
\aliascntresetthe{corollary}

\usepackage{chngcntr}

\counterwithout{figure}{section}
\theoremstyle{definition}
\newaliascnt{definition}{theorem}
\newtheorem{definition}[definition]{Definition}
\aliascntresetthe{definition}

\newaliascnt{remark}{theorem}
\newtheorem{remark}[remark]{Remark}
\aliascntresetthe{remark}

\newcommand*{\R}{\mathbb{R}}
\newcommand*{\N}{\mathbb{N}}
\providecommand{\Z}{\mathbb{Z}}
\renewcommand*{\C}{\mathbb{C}}

\renewcommand{\S}{\mathcal{S}} % Schwartz space
\newcommand*{\FT}{\mathcal{F}}

\newcommand*{\Ball}[1]{\mathcal{B}(#1)}

\newcommand*{\CNormal}[1]{\mathcal{C}\mathcal{N}(#1)}

\renewcommand{\d}{\:\mathrm{d}}

\renewcommand{\phi}{\varphi}
\newcommand*{\eps}{\varepsilon} %{\varepsilonit}

\DeclareMathOperator{\SINC}{sinc}

\DeclareMathOperator{\supp}{supp}

\DeclareMathOperator{\diag}{diag}
\DeclareMathOperator{\cond}{cond\,}
\DeclareMathOperator{\vol}{vol\,}

\renewcommand{\Vec}{\operatorname{vec}}

\DeclareMathOperator{\Var}{var}

\newcommand*{\Expectation}{\mathbb{E}}
\newcommand*{\conj}[1]{\overline{#1}}

\newcommand*{\EXP}{\Expectation\,}

\DeclarePairedDelimiter\abs{\lvert}{\rvert}
\DeclarePairedDelimiter\norm{\lVert}{\rVert}
\DeclarePairedDelimiter\absq{\lvert}{\rvert^2}
\DeclarePairedDelimiter\braces{\lbrace}{\rbrace}
\DeclarePairedDelimiter\brackets{\lbrack}{\rbrack}
\DeclarePairedDelimiter\paren{\lparen}{\rparen}
\DeclarePairedDelimiter\ip{\langle}{\rangle}

\DeclarePairedDelimiterX\aip[1]{\lvert\delimsize\langle}{\rangle\delimsize\rvert}{#1}
\DeclarePairedDelimiterX{\Set}[2]{\lbrace}{\rbrace}{#1 \text{ such that } #2}
\DeclarePairedDelimiterX\setnew[2]{\{}{\}}{#1 \nonscript\;\delimsize|\nonscript\; #2} % NEW
\NewDocumentCommand\sinc{s o g}{% s = star, m = mandatory arg
   \IfBooleanTF{#1}%
   { \SINC\paren*{#3}} % starred
   { \IfNoValueTF{#2}%
     { \IfNoValueTF{#3}%
       { \SINC }% plain
       { \SINC(#3)} % also plain with parenthesis?
     }
     { \SINC\paren[#2]{#3}
     }
   }%
}

\NewDocumentCommand\EU{s o g}{% s = star, o = optional in square brackets, g - optional in curly braces
   \IfBooleanTF{#1}%
   { \EXP\braces*{#3}} % starred
   { \IfNoValueTF{#2}%
     { \IfNoValueTF{#3}%
       { \EXP }% plain
       { \EXP\braces*{#3}} %  plain with parentheses
     }
     { \EXP\braces[#2]{#3}
     }
   }%
}

\newcommand{\E}[1]{\EU*{#1}}

\let\oldchi\chi
\RenewDocumentCommand\chi{s o g}{% s = star, o = optional in square brackets, g - optional in curly braces
   \IfBooleanTF{#1}%
   { \oldchi\paren*{#3}} % starred
   { \IfNoValueTF{#2}%
     { \IfNoValueTF{#3}%
       { \oldchi }% plain
       { \oldchi\paren*{#3}} %  plain with parentheses
     }
     { \oldchi\paren[#2]{#3}
     }
   }%
}

\newcommand*{\epi}[1]{\:e^{2\pi i #1}}
\newcommand*{\empi}[1]{\:e^{-2\pi i #1}}
\newcommand*{\acorr}{autocorrelation\xspace}

\newcommand*{\x}{{\boldsymbol{x}}}
\newcommand*{\y}{\boldsymbol{y}}

\newcommand*{\steta}{{\boldsymbol{\eta}}}

\newcommand*{\h}{{\boldsymbol{h}}}

\renewcommand*{\L}{\Z_L}                                                                %{\{0,\dotsc, L-1\}}

\renewcommand*{\c}{\boldsymbol{c}}

\newcommand*{\OPW}{\operatorname{OPW}}
\newcommand*{\St}{St}
\newcommand*{\StOPW}{\St\!\OPW}

\newcommand*{\rect}{\Box} %{\APLbox} %{\operatorname{rect}}

\newcommand*{\KN}{\sigma}

\newcommand*{\Zak}{\mathcal{Z}}
\renewcommand{\H}{\boldsymbol{H}}
\newcommand*{\Zee}{\boldsymbol{\mathcal{Z}}} %{\tilde{G}_p}

\newcommand*{\GG}{\tensoratom{G}{}{G}{}}

\newcommand{\placeholder}{\mathord{\,\cdot\,}}
\makeatletter\renewcommand*\env@matrix[1][*\c@MaxMatrixCols c]{\hskip -\arraycolsep\let\@ifnextchar\new@ifnextchar\array{#1}}\makeatother
\newcommand*{\e}{\boldsymbol{e}} % temporary name for a discrete stoch proc.

\renewcommand{\phi}{\varphi}

\renewcommand*{\Zee}{\boldsymbol{Z}}
\newcommand*{\Zvec}{\vec{\Zee}}
\renewcommand*{\x}{x}   %input signal
\newcommand*{\stx}{\boldsymbol{\x}}   %input signal
\renewcommand*{\y}{\boldsymbol{y}}   %echo
\newcommand{\Bmax}{B_{\text{\tiny{max}}}}
\newcommand{\Tmax}{T_{\text{\tiny{max}}}}
\newcommand{\Rz}{\vec{R}_{\scriptscriptstyle Z}}

\newcommand{\Rzhat}{\widehat{\vec{\boldsymbol{R}}}_{\scriptscriptstyle Z}}
\newcommand{\Chat}{\widehat{\boldsymbol{C}}}
\newcommand{\Cvec}{\vec{C}}  %was D
\newcommand{\Cfull}{\vec{C}_{\textup{\tiny{full}}}}  %was D
\newcommand{\Kfull}{K_{\textup{\tiny{full}}}}  %was D
\newcommand{\Cvechat}{\widehat{\vec{\boldsymbol{C}}}}

\renewcommand{\e}{\vec{\steta}} % was {\boldsymbol{e}}
\newcommand{\stetavec}{\vec{\steta}} % was {\boldsymbol{e}}
\renewcommand{\KN}{\boldsymbol{\sigma}} %{\mathrm{\mathbf{H}}} %{\mathrm{PSD}}
\renewcommand{\t}{\tau}
\renewcommand{\GG}{\conj{G} \otimes G}
\newcommand{\GGGamma}{\conj{G}_\Gamma \otimes G_\Gamma}
\newcommand{\comment}[1]{}

\begin{document}
\title{\mytitle \\ }
\author{\myauthor
\thanks{G.~E.~Pfander and P.~Zheltov are with Jacobs University Bremen.}%
\thanks{Emails:\{g.pfander, p.zheltov\}@jacobs-university.de}%
\thanks{G.~E.~Pfander and P.~Zheltov acknowledge funding by the Germany Science Foundation (DFG) under Grant 50292 DFG PF-4, Sampling Operators.}
\thanks{\today}
}

\maketitle

\begin{abstract}
In many radar scenarios, the radar target or the medium is assumed to possess randomly varying parts. The properties of a target are described by a random process known as the spreading function. Its second order statistics under the {WSSUS} assumption are given by the scattering function.
Recent developments in operator sampling theory suggest novel channel sounding procedures that allow for the determination of the spreading function given complete statistical knowledge of the operator echo from a single sounding by a weighted pulse train.

We construct and analyze a novel estimator for the scattering function based on these findings. Our results apply whenever the scattering function is supported on a compact subset of the time-frequency plane. We do not make any restrictions either on the geometry of this support set, or on its area.
Our estimator can be seen as a generalization of an averaged periodogram estimator for the case of a non-rectangular geometry of the support set of the scattering function.
\end{abstract}


\section{Introduction}\label{sec:intro}
The classical scenario in a delay-{D}oppler radar system is that a test signal $x(t)$ is reflected off a target $\H$ and the echo $\y(t) = \H\, \x(t)$ is received.%
\footnote{We use boldface to denote quantities that we may assume to be random variables or stochastic processes, $\conj{A}$ for complex conjugation and $A^*$ for conjugate transpose. We denote $\chi(t)$ a characteristic function of interval $[0,1)$ and $\chi_R(t)$ a characteristic function of set $R$.}
The characteristics of the target, such as its possibly stochastic time-varying impulse response and spreading function, or its scattering function in case of a wide sense stationary with uncorrelated scattering (WSSUS) target, have to be reconstructed from the echo $\y(t)$.
Analogously, in a wireless communication setting, a sounding signal $x(t)$ is transmitted through the channel $\H$, and the received signal $\y(t)$ is used to characterize $\H$.

A frequently used general target or channel model is given by the integral operator
\begin{equation}\label{eq:channel}
\y(t)= \H\,x(t) = \iint \steta(\tau, \nu) \: M_{\nu}  T_{\tau} \,\x(t) \d\tau \d\nu,
\end{equation}
where $T_{\tau} x(t) = x(t-\tau)$ is a time-shift operator, $M_{\nu} \x(t) = \epi{\nu t}\x(t)$ is a frequency-shift operator and $\steta$ is the \emph{(stochastic) spreading function} of the target.

If the spreading function can be reconstructed from $\y(t)$, we say that $\H$ is \emph{identifiable} by $\x(t)$.

As postulated in Kailath's and Bello's seminal papers \cite{Kailath,BelloMeas} and later proven in general terms in \cite{PfWal,KozPf,Pfander1} for a deterministic operator, its \emph{spread}---that is, the area of support of the spreading function, $\mu (\supp \steta(\tau,\nu))$--- indicates whether a deterministic operator is identifiable or not. If the area is less than one, then the operator is identifiable.
This result has been obtained for stochastic operators in \cite{PfaZh02,PfaZh03}.

On the other hand, if only the support of $\steta(\tau,\nu)$ is known and its area exceeds one, then due to aliasing effects, the operator cannot be determined from the
received (stochastic) echo $\y(t)=\H\,x(t)$ independently of the choice of the sounding signal $x(t)$.

In some applications, it suffices to determine the second-order statistics of  a zero mean stochastic process $\steta(\tau,\nu)$, that is, its so called covariance function $R(\tau,\nu,\tau',\nu') \triangleq \EXP\{\steta(\tau,\nu)\, \conj{\steta(\tau',\nu')}\}$.
In  \cite{PfaZh02,PfaZh03}, it was shown  that a necessary but not sufficient condition for the identifiability of $R(\tau,\nu,\tau',\nu')$ from the output covariance
\begin{equation}
\label{eq:defn.A}
A(t,t') \triangleq \EU*{\y(t)\, \conj{\y(t')}} = \EU*{\H x(t) \, \conj{\H x(t')}}
\end{equation}
is that $R(\tau,\nu,\tau',\nu')$ is supported on a bounded set of 4-dimensional volume less than or equal to one.

In this paper, we focus on a particular class of such operators, namely on stochastic operators that satisfy the \emph{wide-sense stationarity with uncorrelated scattering} (WSSUS) assumption \cite{BelloChar,VanTrees}. In this case, an operators's stochastic spreading function $ \steta(\t,\nu)$   has zero mean for all $\t,\nu$ and must be uncorrelated in each of the variables, that is,
\begin{equation}\label{eq:def.scat}
R(\tau,\nu,\tau',\nu') = \delta(\t-\t') \: \delta(\nu - \nu') \: C(\t,\nu).
\end{equation}
The function $C(\t,\nu)\geq 0$ is called the \emph{scattering function} of the target. It represents the variances of all individual scatterers.

Reflecting the support condition on $\steta(t,\nu)$ for the identifiability of a stochastic operator $\H$,  a WSSUS  target is commonly referred to as \emph{underspread}, if the support of $C(\t,\nu)$ is  contained in a rectangle $[0,\Tmax] \times [-\Bmax/2, \Bmax/2]$ of area $\Bmax \Tmax$ less than or equal to one, and \emph{overspread} otherwise. In \cite{PfaZh02,PfaZh03} it is shown, though, that the size of the support set, whether enclosed in a rectangle or not, is not at all relevant for the identifiability of WSSUS operators.
Indeed, in the WSSUS case, $R(\tau,\nu,\tau',\nu')$ has distributional support of 4-dimensional volume 0, and $\H$ is therefore identifiable, at least in theory.
In this paper we construct a sounding signal $x(t)$ so that indeed the scattering function  $C(\t,\nu)$ can be recovered from  $A(t,t')$ whenever $C(\t,\nu)$ has bounded support.

Note that previous results addressing the identifiability of operators with large spread commonly use a stochastic input $\stx(t)$, for example, white noise \cite{Kailath}, or stipulate a low-dimensional parametric model on the scattering function \cite{Kay}.
We want to emphasize that the results presented in this paper are based on deterministic inputs, and do not assume any prior information about $C(\t,\nu)$ other than boundedness of its support.

\subsection{Estimation problem}\label{sec:estimation}
The herein addressed problem in radar is to determine $C(\t,\nu)$ from the echo $\y$ \cite{Green,Harmon,Gaarder}.
A classical approach that we will follow is to recognize the scattering function $C(\t,\nu)$ as \emph{power spectral density} of the doubly stationary stochastic \emph{time-variant transfer function}, also known as Kohn-Nirenberg symbol,%
\footnote{To accommodate the usage of distributions, such as the Dirac delta function $\delta(t)$, as components of the sounding signal and of the spreading function, the above integral and the equality \eqref{eq:channel} should be understood weakly. We preserve integral notation for clarity.
A rigorous treatment of functional analytic aspects of our approach can be found in \cite{PfaZh03}.}
\[
\KN(t, f) = \iint \steta(\tau,\nu) \epi{(\tau f - \nu t)} \d\tau \d \nu
\]
with
\begin{equation}\label{eq:KN.channel}
\H \x(t) = \int \KN(t,f) \: \FT \x(f) \epi{ f t } \d f.
\end{equation}
This way, we can estimate $C(\t,\nu)$ using the 2D \emph{averaged periodogram estimator} of $\KN$ by taking an ensemble average over the instances $\KN^{(j)}$ of the 2D Fourier transform of $\KN$ obtained by multiple soundings ($j=1,\dotsc, J$)
\[
\Chat(\t,\nu) = \frac{1}{J} \sum_{j=1}^J \absq{ \FT \KN^{(j)} (t, f)} =  \frac{1}{J} \sum_{j=1}^J \absq{\steta^{(j)}(\t,\nu)}.
\]
Note that $\KN^{(j)}$ (or, equivalently, $\steta^{(j)}$) can be accessed only via realizations of the echo $\y^{(j)}(t) = \H^{(j)}\, x(t)$, where $\H^{(j)}$ corresponds to the $j$-th realization of the channel. We specify our choice of the sounding signal $\x(t)$ below.

\subsection{Channel sounding}\label{sec:channel-sounding}
The technique of using unweighted pulse trains as sounding signals  for identification of channels with rectangular spreads has been introduced by Kailath \cite{Kailath}.

Suppose, for illustration, that the spreading function is supported on a rectangle $[0, \Tmax]\times[-\frac{\Bmax}{2}, \frac{\Bmax}{2}]$ such that $\Bmax \Tmax < 1$, that is, the channel is underspread.
Consider the representation of the channel
\[
\H\x(t) = \int \h(t,\tau)\: \x(t-\tau) \d \tau,
\]
where the \emph{time-variant instant response} $\h(t,\t)$ is an inverse Fourier transform of $\steta(\t,\nu)$ in the frequency variable,
\[
\h(t,\tau) = \int \steta(\tau, \nu) \epi{t\nu} \d \nu.
\]
It is easy to see that $\h$ is supported in $\tau$ on the same interval $[0,\Tmax]$ and $\h$ is $\Bmax/2$-bandlimited in $t$.
Observe that sending a single pulse  $\x_n(t) = \delta(t-n\Tmax)$ at time $n\Tmax$ produces the response
\[
\H \x_n(t) = \int \h(t,\tau) \: \delta(t- n\Tmax - \tau) \d\tau = \h(t, t-n\Tmax),
\]
which does not overlap with the echo from any other such pulse $\x_{n'}(t), n'\neq n$, due to the support restriction of $\h$ in the $\tau$ variable.
In other words, performing such soundings for all $n\in\Z$ is equivalent to a single sounding of $\H$ with an unweighted pulse train\footnote{A Russian letter $\Shah$, pronounced ``shah'', is traditionally chosen to denote a pulse train due to its shape.}
\[
\Shah(t) \triangleq \sum_{n\in\Z} \delta(t-n\Tmax).
\]

For a fixed $t\in [0, \Tmax]$, say, $t=0$, performing such sounding, we obtain regular samples 
\[
\{ \H \x_n(0) \}_{n\in\Z} = \{ \h(0, n\Tmax) \}_{n\in\Z}
\]
of the impulse response $\h(0, \placeholder)$ at a rate $\Tmax$.
Since $\h(0, \placeholder)$ is bandlimited with bandwidth $\Bmax/2 = \frac{1}{2\Tmax}$, we can interpolate missing information about $\h(0, \placeholder)$ using Papoulis sampling theorem in the mean-square sense \cite{Papoulis}.
Clearly, the choice of $t=0$ is arbitrary, and the same procedure allows recovery of the entire $\h(t, \tau)$ on $(t,\tau) \in \R\times [0, \Tmax]$ from the echo to the delta train $\Shah(t)$ with a reconstruction formula \cite{Pfander1,PfWal}
\begin{align*}
\h(t,t+\tau) = \chi(\tau / \Tmax) \sum_{k\in\Z} \big(\H\sum_{n\in\Z} T_{n\Tmax} \delta\big)(\tau+k\Tmax)\, \frac{\sin(\pi \Tmax (t-k))}{\pi \Tmax (t-k)}\,.
\end{align*}
Informally, this shows that the entire class of operators with spreading functions supported on a rectangle $S \triangleq [0,\Tmax] \times [-\Bmax/2, \Bmax/2]$, which coincides with the so-called \emph{operator Paley-Wiener space} $\StOPW(S{\times} S)$, can be identified by the delta train $\Shah(t)$ \cite{PfWal}.

In the WSSUS setting, it is easy to observe that  a generic \emph{non-rectangular} support set $M = \supp C(\t,\nu)$ with an overspread bounding box ($\Bmax\Tmax>1$), no unweighted impulse train $\sum_{k\in\Z} \delta(t-kT)$ would be able to resolve the entire operator Paley-Wiener space $\StOPW(M)$ of operators with spreading functions supported on $M$ \cite{PfWal}.
The echo to any such sounding signal will undergo time- or frequency aliasing that will preclude identification.

We prove that by using weighted impulse trains
$ \sum_{k\in \Z} c_k \, \delta(t- kT)$ introduced in \cite{PfWal} as sounding signals, the aliasing effects can be controlled and reverted in case that the set $\supp C(\t,\nu)$ is bounded.
As such, our algorithm can be seen as a generalization for Kailath sounding in case of $\supp C(\t,\nu)$ non-rectangular and overspread.

Unrealistic properties of the periodic impulse train (such as infinite duration and infinite crest factor) prevent this technique from being immediately put to practice as is.
The value of the technique proposed here lies in making transparent the ways how time-frequency analysis machinery works inside the channel.
For example, effects of applying different filters to the input and output, such as time-gating the sounding signal, bandpass filter to the received signal and changing the pulse shape by convolving it with some smoother kernel, can then be isolated and analyzed separately \cite{KraPfa}.

\subsection{Numerical experiments}
A discretization of the channel, undertaken here for the purposes of digital simulation, can perhaps serve as a starting point for such analysis.

\section{Theoretic results on scattering function identification}\label{sec:target_id}
In this section we show how the compactly supported scattering function $C(\t,\nu)$ can be reconstructed exactly  given complete knowledge of the second order statistics of the echo $\y(t)$ to the delta train sounding signal $\x(t)$.

\subsection{Geometry of the support set \texorpdfstring{$\supp C(\t,\nu)$}{supp C(τ,ν)}}\label{sec:geometry}
Let $ \steta(\t,\nu)$ be the spreading function of a radar target $H$ with the distributional support set $\supp \steta(\t,\nu)$ that lies within the compact set $[0,\Tmax]\times[0, \Bmax]$ with $\Tmax \Bmax=L\in \N$. 
\footnote{Other bounding boxes, for example, a causal symmetric box $[0,\Tmax]\times[-\frac{\Bmax}{2}, \frac{\Bmax}{2}]$, can be accommodated by a simple translation \cite{PfWal}.}
Since $\EXP\absq{\steta(\t,\nu)}=0$ implies $\steta(\t,\nu)=0$ almost surely, in the WSSUS case the support sets of the spreading function and the scattering function coincide.

Let $T \triangleq 1/\Bmax=\Tmax/L$, $B \triangleq 1/\Tmax=\Bmax/L$ and $ R = [0,T) \times [0,B)$, so
\begin{equation}\label{eq:steta.covering.full}
\supp C(\t,\nu) \subseteq \bigcup_{a,b=0}^{L{-}1}  \Bigl[a T, (a+1)T\Bigr) \times \Bigl[b B, (b+1) B\Bigr).
\end{equation}

The assumption \eqref{eq:steta.covering.full} suffices to fully describe the herein proposed channel estimator.

In order to improve the performance of the estimator, the condition $\supp \steta(\t,\nu)\subseteq [0,\Tmax]\times[0, \Bmax]$ may be replaced with a more general condition that $\supp \steta(\t,\nu)$ is contained within a fundamental domain of $\R^2$ under a lattice $\Tmax\Z \times \Bmax\Z$ acting on it by translations.
That is, we can replace  \eqref{eq:steta.covering.full} by the requirement that for some discrete finite set $\Gamma = \{(a_j, b_j)\} \subset \Z^2$ satisfying the anti-aliasing condition
\begin{equation}
\label{eq:anti-aliasing}
\Gamma \cap \Big(\Gamma + (m L, n L)\Big) = \varnothing \quad \text{ for all } m,n \in \Z,
\end{equation}
the support of the scattering function is covered by translates of the prototype rectangle $R = [0, T)\times[0,B)$,
\begin{equation}\label{eq:steta.covering}
\supp C(\t,\nu) \subseteq \bigcup_{(a_j,b_j)\in\Gamma}  \Bigl[a_j T, (a_j+1)T\Bigr) \times \Bigl[b_j B, (b_j+1) B\Bigr),
\end{equation}
Potentially, we may need less than $L^2$ translates of $R$ to cover $\supp \steta (\t,\nu)$, but it immediately follows from \eqref{eq:anti-aliasing} that $\abs{\Gamma}\leq L^2$.
We would then call the support set $\supp  C(\t,\nu)$ \emph{$(R,\Gamma)$-rectified}.
(We simultaneously obtain an $(R,\Gamma)$-rectification of the set $\supp \steta(\t,\nu)$.) See \autoref{fig:rectification} for illustration.
In the following, we shall assume  \eqref{eq:steta.covering}, but advise the reader to first consider the special case \eqref{eq:steta.covering.full}, that is,  $\Gamma=\{0,1,\ldots,L{-}1\}{\times}\{0,1,\ldots,L{-}1\}$.

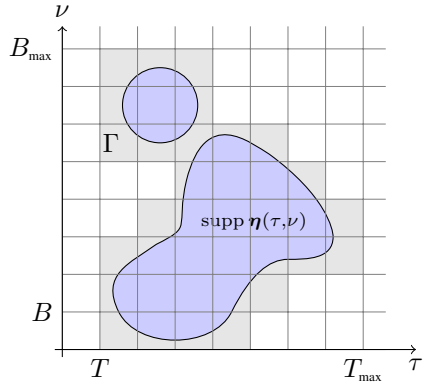
\begin{figure}[htbp]
\centering
% \includegraphics[width=0.5\columnwidth]{tikz-rectification}
% \includestandalone{tikz-rectification}
\begin{tikzpicture}[scale=0.5]
\def\L{9}
\def\Lminusone{8}
\begin{scope}[draw=gray!20, fill=gray!20]      %[pattern=north east lines wide, pattern color=gray,draw=gray]
\filldraw (1,0) -- ++(0,3) -- ++(1,0) -- ++(0,1) -- ++(1,0) -- ++(0,1) -- ++(-2,0) -- ++(0,3) -- ++(3,0) -- ++(0,-2) -- ++(2,0) -- ++(0,-1) -- ++(1,0) -- ++(0,-1) -- ++(1, 0) -- ++(0,-2)  -- ++(-2,0) -- ++(0,-1) -- ++(-1,0) -- ++(0,-1) -- ++(-4,0);
\end{scope}
\begin{scope}[shift={(0,-1)}]
\filldraw[fill=blue!20, draw=black]
                (1.5,2) .. controls (1, 3)  and (1.8,3.3) .. (2.5,3.8)
						.. controls (3.5,4.3) and (3.0,4.2) .. (3.3,5.5)
						.. controls (3.8,7.5) and (5, 6.5) .. (5.5,6.2)
						.. controls (7, 5.1) and (7.2, 4.2) .. (7.2,4)
						.. controls (7.2, 3.4) and (6.5, 3.4) .. (6,3.4)
						.. controls (5.5 ,3.4) and (5,3) .. (4.5,2)
						.. controls (4,1) and (2,1) .. (1.5,2)
						(2.6, 7.5) circle (1);
\node at (5.1,4.4) {$\scriptstyle{\supp \steta(\t,\nu)}$};
\node at (1.3,6.5) {$\Gamma$};
\end{scope}
\draw[style=help lines, step={(1,1)}] (0,0) grid (\L-0.4, \L-0.4);
\draw[->] (-0.2,0) -- (\L+0.4,0) node[below] {$\t$};
\draw[->] (0,-0.2) -- (0,\L-0.4) node[above] {$\nu$};
\node[below] at (1,0) {$T$};
\node[left] at (0,1) {$B$};
\node[below] at (8,0) {$\Tmax$};
\node[left] at (0,8) {$\Bmax$};
\end{tikzpicture}
\caption{Rectification of a scattering function, $L=8$}
\label{fig:rectification}
\end{figure}

For any $(a,b)\in\Gamma$, we define the patch $\steta_{(a,b)}(\t,\nu)$ to be
\begin{equation*}
 \steta_{(a,b)}(\t,\nu) \triangleq \chi_R(\t,\nu) \:  \epi{ b B T}\: \steta(\t+a T,\nu+b B),
\end{equation*}
where we use the phase-adjusting factor $\epi{b B T}$ for later convenience.
We form a column vector of  doubly periodized $\steta(\t,\nu)$
\begin{align}\label{eq:periodized.eta}
\stetavec(\t,\nu) =\brackets[\Big]{  \sum_{m,n\in\Z} \steta_{(a + m L, b+n L)}(\t,\nu) }_{a,b=0}^{L-1},
\end{align}
which reduces to
\begin{align*}
\stetavec(\t,\nu) =\brackets[\Big]{  \steta_{(a , b)}(\t,\nu) }_{a,b=0}^{L-1}
\end{align*}
if $\Gamma=\{0,1,\ldots,L{-}1\}{\times}\{0,1,\ldots,L{-}1\}$.
Note that due to absence of aliasing, each sum in \eqref{eq:periodized.eta} has only a single nonzero term, so $\stetavec(\t,\nu)$ is an $L^2$-long vector consisting of all patches $\steta_{(a,b)}(\t,\nu), (a,b)\in\Gamma$, and zero entries for the rest.
This gives us a decomposition of the spreading function 
\[
 \steta(\t,\nu)= \sum_{(a,b)\in \Gamma} \steta_{(a,b)}(\t-a T,\nu-b B) \empi{  b B T}.
\]
We define $C_{(a,b)} \colon [0,T)\times[0,B) \to \R^{+}$ by
\[
C_{(a,b)}(\t,\nu)  \: \delta(\t-\t') \: \delta(\nu-\nu') \coloneqq \E{ \stetavec_{(a,b)}(\t,\nu) \:  \conj{\stetavec_{(a,b)}(\t',\nu')}},
\]
and let the vector-valued function
\begin{equation}\label{eq:defn.D}
\Cvec(\t,\nu) \coloneqq  \brackets*{ C_{(a,b)}(\t,\nu) }_{a,b=0}^{L-1}.
\end{equation}
Observe now that due to the WSSUS property of the channel, different patches are uncorrelated, that is, the covariance matrix $\E{\stetavec(\t,\nu)\: \stetavec(\t',\nu')^*}$ of the vector $\stetavec(\t,\nu)$ has nonzero entries only on the diagonal\footnote{Following Matlab notation, for any vector $\vec{v} :  \Z^{L^2} \to X$, we define a diagonal matrix $\diag \vec{v}$ to have elements of $\vec{v}$ on the diagonal, and zeros elsewhere, that is,
$(\diag \vec{v})[i,j] = \delta_{ij} \: \vec{v}[i],$ where $\delta_{ij}$ is a Kronecker delta.}
\begin{equation}\label{eq:E.Cvec}
\E{\e(\t,\nu) \: \e(\t',\nu')^*} = \delta(\t-\t')\: \delta(\nu-\nu') \:\diag \Cvec(\t,\nu).
\end{equation}
We now have a scattering function decomposition
\begin{equation}\label{eq:patch}
\begin{split}
C(\t,\nu) &= \sum_{(a,b)\in\Gamma} C_{(a,b) \bmod L} (\t -a T, \nu - b B) \chi_{[0,T)\times[0,B)}(\t-aT, \nu- bB)
\end{split}
\end{equation}
for all $(\t,\nu) \in [0,\Tmax]\times[0,\Bmax]$, and define
\begin{equation}\label{eq:Cfull}
\Cfull(\t,\nu) \coloneqq \Vec \diag \Cvec(\t,\nu),
\end{equation}
where we use the vectorization operation $\Vec \colon \C^{M\times N} \to \C^{MN}$ \cite{VanLoan}
\[
(\Vec{V})_{iN+j} \coloneqq V_{i,j}, \quad i=1,\dotsc, M, \quad  j=1, \dotsc, N,
\]

\subsection{Finite-dimensional Gabor frames}\label{sec:gabor}
Let $c$ be a column vector in $\C^L$  that we will choose later, and let $G = [ c_{(a,b)} ]^{L-1}_{a,b=0}$ be the Weyl-Heisenberg-Gabor frame for $\C^L$, where we define
\[
c_{(a,b)} \triangleq \pi(a,b) c = \pi(\lambda) c,
\]
where the discrete time-frequency shift $\pi(\lambda)$ is defined via
\[
\pi(a, b)c[j]= \epi{r b /L} c[j-a],\ j = 1,\dotsc,L.
\]
The vector $c$ can be chosen in such a way that any $L$-element subset of $G$ is linearly independent, a condition of the frame known as the \emph{Haar property} \cite{LPW, Mal13}
The selection procedure allows us to choose $c$ with elements having modulus one, if necessary.
In fact, choosing the entries of $c$ randomly from a uniform distribution on a unit circle guarantees the Haar property to hold with probability one.
However, it turns out that such choice of $c$ is not optimal for this setting from the numerical viewpoint.
In \autoref{thm:singular} we characterize vectors $c$ which are optimal with respect to their numerical performance.

We define $G_\Gamma$ to be the $L\times \abs{\Gamma}$ submatrix of those columns of $G$ indexed by the covering $\Gamma \bmod L$, that is,
\[
G_\Gamma = \begin{bmatrix}  c_{(a,b) \bmod L}  \end{bmatrix}_{(a,b)\in\Gamma \bmod L}
\]
It is useful to note that it is the choice of the Gabor frame guarantees that  $c_{(a,b) \bmod L} = c_{(a,b)}$ for all $a,b$, not just for $a,b \in \{0, \dotsc, L-1\}$.

As illustrated below, we are interested in the Kronecker product matrix $\GGGamma \colon  L^2 \times \abs{\Gamma}^2 \to \C$, and its submatrix of Kronecker products of individual vectors with themselves
\begin{equation}\label{eq:defn.K}
K \coloneqq (\GG)_{\diag (\Gamma \times \Gamma)} = \begin{bmatrix} \conj{c_{(a,b)}} \otimes c_{(a,b)} \end{bmatrix}_{(a,b)\in\Gamma}.
\end{equation}

The full Kronecker product matrix  $\GG$ does not have the Haar property,%
\footnote{The same definition for the Haar property still makes sense, if we view $\GG$ as a Gabor frame over a non-cyclic abelian group $\Z_L \times \Z_L$ with window $\conj{c} \otimes c$.}
in fact, there always exist column subsets of size $3 L-2$ of it that are linearly dependent).
However, the \enquote{diagonal} geometry of this particular subset of $\Gamma \times \Gamma$ guarantees that $K = (\GG)_{\diag (\Gamma \times \Gamma)}$ is full rank for any $\Gamma$, per the following result.

\begin{theorem}(\cite[Theorem 15]{PfaZh02})\label{thm:diagonal}
Let $G = \{ c_{(a,b)} \}_{a,b=0}^{L-1}$  be a Gabor frame generated by $c \in \C^L$. Then for almost all $c \in \C^L$, the set
\[
\Kfull \coloneqq \braces*{ \conj{\pi(a,b) \, c } \otimes \pi(a,b)\, c}_{a,b \in \Z_L}
\]
of all tensor products of each Gabor element with itself is linearly independent.
\end{theorem}

Moreover, we have the following characterization of vectors $c$ for which the system $\Kfull$ above is in general linear position.
This is in contrast to the much more difficult problem of characterizing $c\in\C^L$ such that the Gabor system $\{ \pi(a,b) c \}_{a,b=0}^{L=1}$ is in general linear position.

\begin{theorem}\label{thm:singular}
 The singular values of the $L^2\times L^2$ matrix $\Kfull$
 are given by the values of the (discrete) short-time Fourier transform $|V_c c(a,b)| \triangleq |\langle c, \pi (a,b)c\rangle|$, $(a,b) \in \Z_L\times \Z_L$.
\end{theorem}

\begin{proof}
 For a vector $v[a,b]_{a,b=0}^{L-1}$ consider the matrix $W \triangleq \sum_{a,b} v_{a,b} \, (\pi(a,b)c)\, (\pi(a,b)c)^\ast$.
We compute
\allowdisplaybreaks[4]
\begin{align*}
\|W\{v_{a,b}\}\|^2_2  &= \sum_{m,n} \big|W\{v_{a,b}\}(m,n)\big|^2  \\
 &=\sum_{m,n} \big|\sum_{a,b} v_{a,b} \, \pi(a,b)c[m]\,\pi(a,-b)\overline{ c[n] }\big|^2 \\
 &=\sum_{m,n} \big|\sum_{a,b} v_{a,b} \, e^{2\pi i b m/L} \, c[m-a] \empi{ b n/L} \, \overline{ c[n-a] }\big|^2 \\
 &=\sum_{m,n} \big|\sum_{a,b} v_{a,b} \epi{b (m-n)/L} \, c[m-a]\,  \conj{ c[n-a]}\big|^2 \\
 &=\sum_{m,n} \big|\sum_{a,b} v_{a,b} \epi{b m/L} \, c[m+n-a]\,  \conj{ c[n-a]}\big|^2 \\
 &=\begin{multlined}[t][0.7\columnwidth]
   \tfrac 1 L \sum_{m,r} \big|\sum_{a,b} v_{a,b} \, \epi{b m/L}  \sum_n c[m+n-a]\,  \overline{ c[n-a]}\, e^{-2\pi i n r/ L}\big|^2
 \end{multlined} \\
 &=\begin{multlined}[t][0.7\columnwidth]
   \tfrac 1 L \sum_{m,r} \big|\sum_{a,b} v_{a,b} \epi{b m/L}  \sum_n c[m+n]\,  \overline{ c[n]}\, e^{-2\pi i (n+a) r/ L}\big|^2
 \end{multlined} \\
 &=\begin{multlined}[t][0.7\columnwidth]
 \tfrac 1 L \sum_{m,r} \big|\sum_{a,b} v_{a,b} \epi{(b m/L-a r/L)} \sum_n c[m+n]\,  \conj{ c[n]}\, e^{-2\pi i n r/ L}\big|^2
\end{multlined} \\
 &=\begin{multlined}[t][0.7\columnwidth]
   \tfrac 1 L \sum_{m,r} \big|\sum_{a,b} v_{a,b} \epi{(b m/L-a r/L)} \sum_n c[n]\,  \conj{ c[n-m]}\, e^{-2\pi i (n-m) r/ L}\big|^2
\end{multlined} \\
   &=\begin{multlined}[t][0.7\columnwidth]
    \tfrac 1 L \sum_{m,r} \big|\Big(\sum_{a,b} v_{a,b} \epi{(b m/L-a r/L)}\Big) \langle c, \pi(m,r) c\rangle \epi{m r/ L}\, \big|^2 
    \end{multlined}  \\
   &= \sum_{m,r} \big|\sum_{a,b} v_{a,b} \,\tfrac 1 {\sqrt{L}} \epi{(b m/L-a r/L)}\big|^2 \ \big|\langle c,\pi(m,r)c \rangle  \big|^2.
\end{align*}
As $\braces*{\tfrac 1{\sqrt{L}}\ \epi{(b m/L-ar/L)}}_{(a,b) \in \Z^2}$ is an orthonormal basis for $\mathbb C^{L\times L}$, the result follows.
\end{proof}

\begin{corollary} \label{cor:condition-number-fiducial}
The condition number of the $L^2\times L^2$ matrix $\Kfull$
is bounded below by $\sqrt{L+1}$. The lower bound  is achieved if and only if $\{\pi(a,b)\,c\}$ is an equiangular frame, that is, if and only if $c$ is a fiducial vector.
\end{corollary}

\begin{proof}
\autoref{thm:singular} implies that the smallest and largest singular values are given by
$\lambda_{\text{min}} = \min_{a,b} \big|\langle c,\pi(a,b)c\rangle\big|^2$,
and respectively,
$ \lambda_{\text{max}} = \max_{a,b} \big|\langle c,\pi(a,b)c\rangle\big|^2$
$=\norm{c}^2$.
We conclude that
$\{ \pi(a,b)c \otimes \pi(a,b)c^\ast \}_{a,b \in  \Z_L}$ is a Riesz basis if and only if $V_c c[a,b]=\langle c,\pi(a,b)c\rangle$ never vanishes.

To estimate the smallest singular value, we compute
\begin{align*}
 \|V_c c\|_2^2 & = \sum_{a,b}\big|V_c c[a,b]\big|^2                                                                        \\
               & = \sum_{a,b}\big|\sum_k c[k] \overline{ c[k-a]} e^{-2\pi i k b/L}\big|^2                                   \\
               & = \tfrac 1 L \sum_{a,r}\big| \sum_k c[k] \,\conj{ c[k-a]} \, \sum_b \empi{(k-r)b/L}\big|^2              \\
               & = \tfrac 1 L \sum_{a,r}\big| \sum_k c[k] \, \conj{ c[k-a]} \, L \, \delta(k-r)\big|^2                        \\
               & =  L \sum_{r}\sum_{a}\big| c[r] \, \conj{ c[r-a]}\big|^2 = L \norm{c}^4.
\end{align*}
Clearly, the $\min_{a,b} |V_c c[a,b]|$ is going to be the largest if $V_c c$ is constant with the exception of a point $(0,0)$.
This holds if $\{\pi(a,b) c \}$ is an equiangular frame, namely, if $c$ is a scalar multiple of a fiducial vector \cite{blanchfield2013orbits, appleby2007symmetric, appleby2005symmetric, Ras13,Waldron12}.
Assuming without loss of generality that $\norm{c}_2=1$, we have that $\min_{a,b} |V_c c[a,b]|=\max_{(a,b)\neq (0,0)} |V_c c[a,b]|$ achieves the Welch bound $\sqrt{\frac {L-1}{L^2-1}}=1/\sqrt{L+1}$.
The condition number is the ratio of largest and smallest singular values, that is,
\[
\cond \Kfull = \frac{\lambda_{\max}}{\lambda_{\min}} = \frac 1 {\frac{1}{\sqrt{L+1}}}=\sqrt{L+1}.\qedhere
\]
\end{proof}

\subsection{Zak transform and scattering function identification}\label{sec:zak}
We define the \emph{non-normalized Zak transform}  $\Zak \colon L^2(\R) \to L^2\left([0, L T]\times [0,B]\right)$ by
\begin{equation*}\label{eq:defn.Zak}
\Zak y(\t,\nu) \coloneqq \sum_{n \in \Z} y(t-n L T)\epi{n L T \nu}.
\end{equation*}
Consider the image $\Zak \y(\t,\nu)$ of the Zak transform of the echo $\y(t) = \H \x(t)$, defined on $(\t,\nu) \in [0, JT)\times[0,B)$, and form a column vector $\Zvec(\t,\nu)$ of $T\times B$ patches of it, again with a convenient phase-adjusting factor and a normalization.
Namely, for $(\t,\nu) \in [0,T)\times [0,B)$ and $p = 0, \dotsc, L-1$, we let
\begin{gather}
\Zee_p(\t,\nu)  \coloneqq B^{-1} \empi{\nu (\t + p T)} \: \Zak \y(\t+ p T,\nu),\\
\Zvec(\t,\nu) = \Bigl[ \Zee_p(\t,\nu) \Bigr]_{p=0.}^{L-1}
\label{eq:Zee_p}
\end{gather}

In \cite{PfWal} we exploit the tight connection of the Zak transform with Gabor frame theory to show
\begin{equation}\label{eq:mixing}\begin{split}
\Zvec(\t,\nu) = G \, \stetavec (\t,\nu),  
\end{split}\end{equation}
a possibly underdetermined system of equations.
Here, we compute for any  $(\t,\nu),(\t',\nu') \in [0,T)\times[0,B)$,
\[
\E{\Zvec(\t,\nu) \: \Zvec(\t',\nu')^*} = G \: \E{\e(\t,\nu) \,  \e(\t',\nu')^*} \,G^*,
\]
or in vectorized form, by \eqref{eq:E.Cvec} and \eqref{eq:defn.K},
\begin{align*}\label{eq:vec.ZZ}
\E{\Zvec(\t,\nu) \otimes \conj{\Zvec(\t',\nu')}} 
&= (\GG) \: \E{\stetavec(\t,\nu) \otimes \conj{\stetavec(\t',\nu')}}\\
&= \delta(\t-\t')\: \delta(\nu-\nu') \: (\GG)\:  \Cfull(\t,\nu) \\ 
&= \delta(\t-\t')\: \delta(\nu-\nu') \: \Kfull \: \Cvec(\t,\nu),
\end{align*}
where due to sparsity of $\Vec \diag \Cvec(\t,\nu)$ we can restrict the coefficient matrix $\GG$ to its submatrix $\Kfull$, or, in case of an even more meager support set, to $K = (\GG)_{\diag \Gamma \times\Gamma}$.

Since the support of the right-hand side equals the suppoer of the left-hand side, the autocorrelation of the Zak transform of the output $\Zak\,\y$ is described by
\begin{equation}\label{eq:Rz}
\Rz(\t,\nu) = K\: \Cvec(\t,\nu)
\end{equation}
for some observable column vector of functions $\Rz(\t,\nu)$ of length $L^2$ such that
\[
\Rz(\t,\nu) \: \delta(\t-\t')\: \delta(\nu-\nu') \coloneqq  \E{\Zvec(\t,\nu) \otimes \conj{\Zvec(\t',\nu')}}.
\]
Since by \autoref{thm:diagonal}, $K$ is left invertible, we can recover the vector $\Cvec(\t,\nu)$ of patches of the scattering function, and hence, the scattering function itself, by inversion
\begin{equation}\label{eq:C.Rz}
\Cvec(\t,\nu) = K^{-1} \: \Rz(\t,\nu),
\end{equation}
The above arguments lead to the following result.
\begin{theorem}\cite[Theorem 13]{PfaZh02}. \label{thm:Rz}
Let $\H$ be any WSSUS channel such that its scattering function $C(\t,\nu)$ has compact support,
rectified in a such way that \eqref{eq:steta.covering} holds for some rectangle $[0,T)\times[0,B)$.
There exists a vector $c \in \C^L$ for $L=\frac{1}{BT}$ such that the scattering function can  be reconstructed from the autocorrelation of the Zak transform of the output $\y(t) = \H \, \x(t)$ to the input given by the weighted impulse train $\x(t) = \sum_{k\in \Z} c_{k \bmod L} \: \delta(t-k T)$
using \eqref{eq:defn.D} and \eqref{eq:C.Rz}.
\end{theorem}

\begin{remark}\label{rem:Kfull-or-not-Kfull}
The geometry of the support set dictates the rectification $(R,\Gamma)$ which in turn determines the matrix $K$ that needs to be inverted in \eqref{eq:C.Rz}.
Although the full matrix $\Kfull$ is shown to be invertible for $c$ fiducial, in the situation when the support is known, and has small area relative to the bounding box (that is, the cardinality of the index set $\Gamma$ is smaller than the maximum possible $L^2$), we might have better control of the condition number of the matrix $K$, which is a column submatrix of $\Kfull$.
Alternatively, computing a left inverse of $K$ might be numerically cheaper than a full matrix inversion.

In any case, the weight vector $c$ should be chosen in a way as to guarantee a beneficial inversion regime, for example, a fiducial vector.
Alternatively, the weight vector $c$ chosen componentwise independently uniformly at random from the unit circle permits identification with probability 1.
\end{remark}

\begin{remark}
An alternative way to reconstruct $C(\t,\nu)$ is given in the following theorem, proven in a preceding paper \cite{OPZ}.
It requires the area of the support set to be less than or equal to one.
If this criterion is met, it may be of interest, as it employs a simpler cross-correlation $\EXP\{ \Zak y(\t,\nu) \conj{y(t)}\}$ for the recovery of $C(\t,\nu)$.
\begin{theorem}\label{thm:Zyy}
With the same notations as above, let $\H$ be a WSSUS channel with its scattering function $C(\tau, \nu)$ compactly supported on a set of area smaller than one, and rectified in a sense of \eqref{eq:steta.covering} with $L = 1/ (B T)$.
There exists a vector $c\in \C^L$ such that $C(\tau, \nu)$ can be identified from the received echo $\y(t)$ to the weighted impulse train $x(t)$ using
\begin{equation}\label{eq:Zyy}
\delta(t-\t)\:C(\t,\nu) = \tfrac{1}{B} \: A_c^{-1} \: G_\Gamma^{-1} \: \E{  \Zvec(\t,\nu) \: \conj{ \y(t)}},
\end{equation}
where $A_c\in \C^{L\times L} $ is a non-singular diagonal matrix $A[i,j] = c_{-a_j} \delta_{ij}$ with elements of $c$ on the diagonal ordered corresponding to the enumeration of the covering $\Gamma = \{ (a_j, b_j) \}_{j=1}^L$.
\end{theorem}
\end{remark}

\section{Scattering function estimation}\label{ssec:scat_func_est}
In theory, \eqref{eq:C.Rz} and \eqref{eq:Zyy} enable us to reconstruct the scattering function perfectly from the stochastic process $\y(t)$.
But in practice we do not possess exact values of the statistic $\EU[\normal]{ \y(t) \: \conj{\y(t')}}$.
We use ensemble averaging to obtain estimates of the second order statistics of the output.
We note that periodic delta trains have infinite duration and bandlimitation, which presents a challenge to their use for multiple reasons.
The effects of time-gating the signal duration, or using alternative pulse shapes, are explored in the deterministic setting in \cite{KraPfa}.

In our simulations, we deal with periodic delta trains by time-gating the input signal with a sufficiently large window.
We discuss the numerical analysis of our simulation procedures later in \autoref{ssec:discretization}.

\subsection{Estimator construction}\label{sec:estimator-contruction}
To estimate $C(\t,\nu)$, we sound an ensemble of channels with the chosen input signal 
\[
\x(t) = \sum_{k\in\Z} c_k \: \delta(t-k T)
\]
to obtain  $J$ samples of the channel output $\y^{(1)}(t), \dotsc, \y^{(j)}(t)$ corresponding to independent identically distributed samples of the spreading function $\steta^{(j)}(\t,\nu)$.
For every $j=1,\dotsc, J$, we replace the random process $\y(t)$ in \eqref{eq:Zee_p} with samples $\y^{(j)}(t)$ and obtain patches $\Zee_p^{(j)}(\t,\nu)$ for $p=0,\dotsc, L-1,$
\begin{equation}\label{eq:Zee_p^ell}
\Zee_p^{(j)}(\t,\nu)  = B^{-1} \empi{\nu (\t+ p T)} \: \Zak \y^{(j)}(\t+ p T,\nu).
\end{equation}
Out of these, we assemble column function vectors $\Zee^{(j)}(\t,\nu)$ according to  \eqref{eq:Zee_p}.
We estimate the autocorrelation $\Rz(\t,\nu,\t',\nu')$ by averaging
\begin{equation}\label{eq:Dhat}
\Rzhat(\t,\nu,\t',\nu') \coloneqq \frac{1}{J} \sum_{j = 1}^J  \Zvec^{(j)}(\t,\nu) \otimes \conj{\Zvec^{(j)}(\t',\nu')}
\end{equation}
and
\begin{equation}\label{eq:Chat}\begin{split}
\Cvechat(\t,\nu,\t',\nu') &= K^{-1} \: \Rzhat(\t,\nu,\t',\nu') = \frac{1}{J}\sum_{j=1}^J \e^{(j)}(\t,\nu) \otimes \conj{\e^{(j)}(\t',\nu')}.
\end{split}\end{equation}

Clearly, the pointwise statistical properties of the column vector $\Cvechat$ coincide with the properties of an estimator for the scattering function
\[
\Chat(\t,\nu) = \sum_{j=1}^L \Chat_{j,j}(\t - a_j T, \, \nu - b_j B,\,  \t - a_j T,\, \nu - b_j B)
\]
obtained by translating the entries of $\Cvechat$ as in \eqref{eq:patch}.

\subsection{Bias-variance analysis}\label{sec:bias-variance-analysis}
We set $\xi_i \triangleq (\t_i, \nu_i), i=1,\dotsc,4$.
By linearity, from \eqref{eq:Chat} and \eqref{eq:Rz} it follows
\begin{align*}
\EXP \Bigl\{ \Cvechat(\xi_1,\xi_2) \Bigr\} &= \frac{1}{J}\sum_{j=1}^J K^{-1} \: \E{  \Zee^{(j)}(\xi_1)\otimes\conj{\Zee^{(j)}(\xi_2)}} = \Cvec(\xi_1) \: \delta(\xi_1-\xi_2),
\end{align*}
which means that the estimator $\Chat$ is an unbiased estimator of $C$ up to a constant $\delta(\xi_1-\xi_2)$ factor.

We estimate the variance under the simplifying assumption that all individual scatterers $ \steta(\t,\nu)$ are jointly Gaussian complex random variables, and in addition have the \emph{circular symmetry} property, that is, for any $\phi \in \R$, the random variables $e^{i\phi} \steta(\t,\nu)$ have the same distribution as $\steta(\t,\nu)$.

We compute first
\begin{align*}
\E{ \Zee^{(j_1)}(\xi_1)  \otimes\conj{\Zee^{(j_1)}(\xi_2)}  \otimes\conj{\Zee^{(j_2)}(\xi_3)} \otimes \Zee^{(j_2)}(\xi_4)}.
\end{align*}
We start with the case $j_1 = j_2$, omitting the superscripts entirely.
\begin{align*}
\E{ \Zee(\xi_1) \otimes\conj{\Zee(\xi_2)} \otimes \conj{\Zee(\xi_3)} \otimes \Zee(\xi_4)} =(G \otimes \conj{G} \otimes \conj{G} \otimes G) \:  \E{\stetavec(\xi_1) \otimes \conj{\stetavec(\xi_2)} \otimes \conj{\stetavec(\xi_3)} \otimes \stetavec(\xi_4)}
\end{align*}
By Isserlis' moment theorem,
\begin{align*}
\MoveEqLeft \E{\stetavec(\xi_1) \otimes \conj{\stetavec(\xi_2)} \otimes \conj{\stetavec(\xi_3)} \otimes \stetavec(\xi_4)} \\
&=  \Cfull(\xi_1)\: \delta(\xi_1 - \xi_2) \otimes \Cfull(\xi_3)\: \delta(\xi_3 - \xi_4) + \Cfull(\xi_1)\: \delta(\xi_1 - \xi_3) \otimes\Cfull(\xi_2)\: \delta(\xi_2 - \xi_4) \\
\shortintertext{and we shorthand each term}
&\coloneqq  I_1 + I_2.
\end{align*}
Observe that the sparsity condition $(\GG)\: \Cfull (\t,\nu) = K \: \Cvec (\t,\nu)$ implies
\begin{equation}\label{eq:I_1}\begin{multlined}
(G \otimes \conj{G} \otimes \conj{G} \otimes G) \: I_1 = (\conj{K} \otimes K)\: [\Cvec(\xi_1) \otimes \Cvec(\xi_3)] \: \delta(\xi_1 - \xi_2) \: \delta(\xi_3 - \xi_4)
\end{multlined}\end{equation}
and similarly for $I_2$.

The case $j_1 \neq j_2$ is trivial, since by independence of $\Zee^{(j)}(\t,\nu)$ for different $j$'s,
\begin{align*}
\MoveEqLeft \EU{\Zee^{(j_1)}(\xi_1)\otimes\conj{\Zee^{(j_1)}(\xi_2)}} \otimes \EU{\conj{\Zee^{(j_2)}(\xi_3)} \otimes \Zee^{(j_2)}(\xi_4)}\\
&=(G \otimes \conj{G} \otimes \conj{G} \otimes G) \: \brackets*{ \Cfull(\xi_1) \otimes \Cfull(\xi_3) } \: \delta(\xi_1 - \xi_2) \: \delta(\xi_3 - \xi_4) \\
&=(\conj{K} \otimes K) \: \brackets*{ \Cvec(\xi_1) \otimes \Cvec(\xi_3)} \: \delta(\xi_1 - \xi_2) \: \delta(\xi_3 - \xi_4).
\end{align*}

We can now gather all of the above equations and calculate
\allowdisplaybreaks
\begin{align*}
\MoveEqLeft[1] \Vec\,\Var{\Chat(\xi_1, \xi_2,\xi_3,\xi_4)}\\
&= \EU[\Big]{\Cvechat(\xi_1,\xi_2)\otimes \conj{\Cvechat(\xi_3,\xi_4)}} - \EU[\Big]{\Cvechat(\xi_1,\xi_2)}\otimes \conj{\EU[\Big]{\Cvechat(\xi_3,\xi_4)}} \\
&= (\conj{K}^{-1} \otimes K^{-1}) \frac{1}{J^2} \Bigg\lparen \sum_{j_1= j_2=1}^J + \sum_{j_1\neq j_2 = 1}^J \EU[\Big]{\Zvec^{(j_1)}(\xi_1) \otimes \conj{\Zvec^{(j_1)}(\xi_2)} \otimes \conj{ \Zvec^{(j_2)}(\xi_3)} \otimes \Zvec^{(j_2)}(\xi_4)}\Bigg\rparen \\
&\qquad - \brackets*{ \Cvec(\xi_1) \otimes  \Cvec(\xi_3) }\: \delta(\xi_1-\xi_2)  \: \delta(\xi_3-\xi_4) \\
&= (\conj{K}^{-1} \otimes K^{-1}) \frac{1}{J} \Big\lparen (\conj{ \conj{G} \otimes G) } \otimes (\conj{G} \otimes G) (I_1+I_2) + (J-1) (\conj{K} \otimes K ) \: \brackets*{ \Cvec(\xi_1) \otimes \Cvec(\xi_3) } \: \delta(\xi_1 - \xi_2) \: \delta(\xi_3 - \xi_4) \Big\rparen \\
&\qquad  -   \brackets*{ \Cvec(\xi_1) \otimes  \Cvec(\xi_3) }\: \delta(\xi_1-\xi_2)  \: \delta(\xi_3-\xi_4) \\
\shortintertext{which, by \eqref{eq:I_1}, becomes}
&= \frac{1}{J}\Bigg\lparen \brackets*{\Cvec(\xi_1) \otimes \Cvec(\xi_3)} \,\delta(\xi_1 - \xi_2) \: \delta(\xi_3 - \xi_4) \\
&\qquad + \brackets*{\Cvec(\xi_1) \otimes \Cvec(\xi_2)}\: \delta(\xi_1 - \xi_3) \: \delta(\xi_2 - \xi_4)  + (J-1)  \: \brackets*{ \Cvec(\xi_1) \otimes \Cvec(\xi_3)} \: \delta(\xi_1 - \xi_2) \: \delta(\xi_3 - \xi_4) \Bigg\rparen \\
&\qquad  -   \brackets*{ \Cvec(\xi_1) \otimes  \Cvec(\xi_3) }\: \delta(\xi_1-\xi_2)  \: \delta(\xi_3-\xi_4).
\end{align*}
Cancel like terms, and take absolute value.
Only the second summand remains,
\begin{align*}
\MoveEqLeft \abs*{\Vec\Var{\Chat(\xi_1, \xi_2,\xi_3,\xi_4)}} \leq \frac{1}{J} \abs*{\Cvec(\xi_1) \otimes \Cvec(\xi_2)}\: \delta(\xi_1 - \xi_3) \: \delta(\xi_2 - \xi_4).
\end{align*}
This gives us a pointwise estimate. We now integrate over $R^4 = ([0,T)\times[0,B))^4$ to obtain the following bound for the average variance of $\widehat{C_j}$,
\begin{equation}\begin{split}
\frac{1}{\vol R^4} \norm[\Big]{\Var \Cvechat  (\xi_1, \xi_2, \xi_3, \xi_4) }_{L^1} &= \frac{1}{\vol R^4} \norm[\Big]{\Vec\norm{\Var{\Chat(\xi_1, \xi_2,\xi_3,\xi_4)}}_{L^1} }_{\ell^1}\\
&\leq \frac{(BT)^2}{(BT)^4 J} \norm{C(\t,\nu)}_{L^1(\R^2)}^2 = \frac{2L^2}{J} \, \norm{C(\t,\nu)}_{L^1(\R^2).}^2
\label{eq:var.bound}
\end{split}\end{equation}
Hence, the variance of the scattering function estimator $\Chat$ decays inversely proportional to the number of soundings $J$, which is confirmed by the simulation results seen on \autoref{fig:MSE.two}.

\section{Simulation results}\label{sec:simulation}
According to \autoref{thm:Rz}, it is possible to recover the continuous-time scattering function $C(\t,\nu)$ from the \acorr of the output $A(t,t')$.
To test this proposition numerically, we need to address two issues.

First, the sounding signal proposed in \autoref{thm:Rz} is a generalized function with infinite duration.
In fact, to identify an arbitrary channel in $\StOPW(S)$ in full, any identifier necessarily cannot decay either in time or in frequency, as shown in the special case of  deterministic operators in \cite{KraPfa}.
This shortcoming can be resolved if we restrict our interest to the operators' actions only on signals that are approximately time-frequency localized to a set of the time-frequency plane. 
For simplicity, we consider discretization of deterministic channels first.

Recall that a Gabor system $(g, \alpha \Z \times \beta \Z) \triangleq \{M_{\beta n}  T_{\alpha m}g\}_{m,n\in\Z}$ is a \emph{Parseval} frame if for all square integrable signals $x(t)$ we have
\[
\sum_{m,n\in\Z } \abs{\ip{x, M_{\beta n}  T_{\alpha m} g}}^2 = \int |x(t)|^2\,dt.
\]

\begin{definition}\label{defn:approx.tf.localized}
\cite[Definition 3.2]{KraPfa} Let $(g, \alpha \Z \times \beta \Z)$ be a Parseval frame with $g\in \S(\R)$ a function from a Schwartz space of rapidly decreasing functions.
A square integrable signal $x(t)$  is $\eps$-\emph{time-frequency localized} to a compact set $S$, if 
\[
\sum_{\substack{m,n\in\Z\\ (\alpha m,\beta n)\in  S}} \hspace{-.5cm} \abs{\ip{x,  M_{\beta n}  T_{\alpha m}g}}^2 \geq (1-\eps) \sum_{m,n\in\Z } \abs{\ip{x, M_{\beta n}  T_{\alpha m} g}}^2.
\]
\end{definition}

The following result establishes a criterion for two operators being almost indistinguishable in their action on signals time-frequency localized to a set $S$.

\begin{theorem}\cite[Theorem 3.3]{KraPfa} \label{thm:tf-localization} For a Parseval frame $(g, \alpha \Z \times \beta \Z)$ with a Schwartz class window $g$, there exists $C>0$ and for any precision level $\eps>0$ an $r=r(\epsilon)>0$  such that whenever there exists a region $S$ for which
\begin{align*}%\label{eqn:condition1}
\sup_{(t,f)\in \R^2} \abs{\sigma(t,f)} \leq \mu, \quad  \sup_{ (t,f)\in \R^2 } \abs{\widetilde\sigma(t,f)}\leq \mu
\end{align*}
and
\begin{align*}%\label{eqn:condition2}
\sup_{(t,f) \in S+\Ball{r(\eps)}} \abs{\sigma(t,f)-\widetilde\sigma(t,f)} \leq \eps \mu,
\end{align*}
then for any  signal $x\in L^2(\R)$ that is $\eps$-time-frequency localized to $S$, we have 
\[
\norm{H x-\widetilde H x}_{L^2(\R)} \leq C\eps\mu \norm{x}_{L^2}. 
\]
where $H$ and $\widetilde{H}$ are channels with $\sigma$ and $\widetilde{\sigma}$ as their Kohn-Nirenberg symbols, respectively. 
\end{theorem}

In most applications we are only interested in the operator's action within a given time segment and a given frequency band. Hence, the above allows
us to replace an operator in $OPW(M)$ by an approximant in a finite dimensional space. Moreover, the result allows us to replace the identifier, a
tempered distribution, by a smooth square integrable function, as we shall see in the following.
 
The restriction to a finite-dimensional set of scattering functions and the attempt to recover $\steta(\tau,\nu)$ or $C(\t,\nu)$ from a finite length
vector obtained from the \acorr of the response $\y(t) = \H x(t)$ is common in the literature.  Frequently, the compact support condition of
$\eta(\tau,\nu)$ is put aside and a bandlimitation is introduced instead in order to justify a representation of $\eta(\tau,\nu)$ based on regular
samples, which are then taken on the original support of $\eta(\tau,\nu)$. The results in \cite{KraPfa} allow us to provide below a thorough
theoretical justification for this approach; they illustrate that the high frequency content of $\eta(\tau,\nu)$ on its support is only relevant if we
are interested in the operators' actions on a large segment of the time-frequency plane.

%%%%%%%%%%%%%%%%%%%%%%%%%%%%%%%%%%%%%%%
%%%%%%%%%%%%%%%%%%%%%%%%%%%%%%%%%%%%%%% 	DISCRETIZATION
%%%%%%%%%%%%%%%%%%%%%%%%%%%%%%%%%%%%%%%
\subsection{Discretization of the target}\label{ssec:discretization}
\newcommand{\Mup}{\mathsf{M}}
\newcommand{\Nup}{\mathsf{N}}

In the following, we consider the action of $H\in \OPW( [0, T ] \times [-\tfrac B 2,\tfrac B 2])$ on functions $\eps$-time-frequency localized on $S=[0, U ]\times [-\frac V 2, \frac V 2]$.  
It is important to keep track of the distinction between the spreading support set limitation $[0,T]\times [-\tfrac{B}{2}, \tfrac{B}{2}]$ that in applications would come from the properties of the targets, and the time-frequency localization $[0, U], \times [-\tfrac{V}{2}, \tfrac{V}{2}]$, which would be determined by the physical restrictions on the sounding facilities.

Computations in this section do not require $T B=\tfrac 1 L$ with $L\in \N$; in fact, $T B$ may be of arbitrary
size, for example, we can pick $T = \Tmax, B=\Bmax$, with $\Tmax\Bmax=L\in\N$.%  
\footnote{Later, in order to apply operator sampling theory, it shall become useful to consider $H$ to be of the form $H=H_1+ \cdots +H_L$, where $H_j\in \OPW( [k_jT,T+k_jT ] \times [-\tfrac B 2+\ell_j B,\tfrac B 2+\ell_j B])$, with $TB=\tfrac 1 L$.
The following discretization arguments generalize trivially to this setting.}

Whenever the scattering function $\eta(\t,\nu)$ is compactly supported, indicated by $H\in \OPW( [0, T ] \times [-\tfrac B 2,\tfrac B 2])$, the
operator's Kohn-Nirenberg symbol $\sigma$ is bandlimited, admitting the expansion
\begin{align*}
%\label{eq:eta-discretized}
\sigma(t,f) &=   e^{-\pi iT f} \sum_{m\in\Z}  \sum_{n\in\Z} \sigma(\tfrac m B,\tfrac n T) \sinc (B t - m)\, \sinc (T f - n).
\end{align*}

\newcommand{\chisym}{\overline{\chi}}
\newcommand{\Ueps}{\mathcal{U}_\eps}
\autoref{thm:tf-localization} shows that $H$ is almost indistinguishable from $\widetilde H$ with Kohn-Nirenberg symbol 
\begin{align*}
\widetilde \sigma(t,f) =   e^{-\pi iT f } \, \chisym(\tfrac f {V_\eps}) \sum_{\tfrac m B \in U_\eps}  \sum_{n\in\Z} \sigma(\tfrac m B,\tfrac n T) \sinc (B t - m)\, \sinc (T f - n)
\end{align*}
for functions $\eps$-time-frequency localized on $S=[0, U ]\times [-\frac V 2, \frac V 2]$.  Here, $V_\eps \triangleq V+2 r (\eps)$ and $\Ueps \triangleq
[-r(\eps), U+r(\eps)]$ where $r=r(\eps)$ is given by \autoref{thm:tf-localization}.  Possibly enlarging $V_\eps$ allows us to assume
$V_\eps/B=K\in\N$.

Observe that 
\begin{align*}
\widetilde H x(t) &= \int  \widetilde \sigma(t,f)\, \widehat x(f)\,e^{2\pi i t f} \d f, \\
\widetilde \sigma(t,f) &=   e^{-\pi iT f } \chi(\tfrac f {V_\eps}) \sum_{\tfrac m B \in \Ueps}  \sum_{n\in\Z} \sigma(\tfrac m B,\tfrac n T) \sinc (B t - m)\, \sinc (T f - n).
\end{align*}
implies that $\widetilde Hx(t)=\widetilde H  (V_\eps\sinc(V_\eps \, \cdot)\ast x)(t)$ for all $x$, so we can assume that all input signals $x$ are bandlimited to  $[-\frac {V_\eps} 2, \frac {V_\eps} 2]$. Moreover, since we are only  only interested to study the operator's action for functions localized in time to  $[0,U]$, it suffices to consider the operator on signals of the form
\begin{equation}
\label{eq:tf-localized-signal}
x(t) =   \sum_{\tfrac k {V_\eps}\in[0,U]}  x(\tfrac{k}{V_\eps}) \, \sinc{V_\eps t -  k}.
\end{equation}
 
Taking an inverse Fourier transform in $f$ implies that the time-varying impulse response of $\widetilde H$ is given by
\begin{align*}
\widetilde h(t,\tau) =  \sum_{\tfrac{m}{B} \in \Ueps}   \sum_{n\in\Z} \widetilde h (\tfrac{m}{B},\tfrac{n}{V_\eps}) \sinc (B t - m)\, \sinc (V_\eps \tau - n).
\end{align*}
Oversampling gives us the alternative expansion 
\begin{align*}
\widetilde{h}(t,\tau) = 
\sum_{\tfrac{m}{V_\eps} \in \Ueps} \sum_{n\in\Z} \widetilde h (\tfrac{m}{V_\eps},\tfrac{n}{V_\eps}) \sinc ({V_\eps} t - m)\, \sinc (V_\eps \tau - n).
\end{align*}

Let us now compute the echo of the target with the instant response function $\widetilde h$ to $x(t)$ of the kind \eqref{eq:tf-localized-signal},
\[
y(t) = \widetilde H   x(t) = \int \widetilde h(t, \t)\,   x (t-\t) \d\t,
\]
that is, 
\begin{align*} 
y(t)                    & =  \adjustlimits \sum_{\tfrac m B \in \Ueps} \sum_{\tfrac k {V_\eps}\in[0,U]} \sum_{n\in\Z} \widetilde h (\tfrac m  B,\tfrac n {V_\eps})\, x(\tfrac k {V_\eps})\, \sinc ( B t - m)\,\int \sinc (V_\eps (t - \t) - r) \sinc (V_\eps \tau - n) \d \t
\end{align*}
Note that on $\frac 1 {B} \Z$, we have due to orthogonality of the $\sinc$s,
\begin{align*}\label{eq:output}
y(\tfrac \ell {B})      & = \adjustlimits \sum_{\tfrac{m}{B} \in \Ueps} \sum_{\tfrac{k}{V_\eps}\in[0,U]} \sum_{n\in\Z}  \widetilde h (\tfrac{m}{B},\tfrac{n}{V_\eps})\, x(\tfrac k {V_\eps})\, \sinc ( \ell - m)\, \int \sinc (\ell K -V_\eps  \t - k) \sinc (V_\eps \tau - n) \d \t                                                                                                               \\
                        & = \sum_{\tfrac m B \in \Ueps} \sum_{\tfrac k {V_\eps}\in[0,U]} \sum_{n\in\Z}  \widetilde h (\tfrac{m}{B},\tfrac n {V_\eps})\, x(\tfrac k {V_\eps})\, \delta_{ \ell - m}\, \delta_{\ell K-k -n} \\
                        & =\sum_{\tfrac k {V_\eps}\in[0,U]}  \widetilde h (\tfrac {\ell }  {B},\tfrac {\ell K-k} {V_\eps}) \; x(\tfrac k {V_\eps})\,  
\end{align*}
for $ \tfrac \ell B \in \Ueps$, and 0 else, and, based on the oversampling alternative,
\begin{align*} 
y(\tfrac \ell {V_\eps}) & =\sum_{\tfrac k {V_\eps}\in[0,U]}  \widetilde h (\tfrac {\ell }  {V_\eps},\tfrac {\ell -k} {V_\eps}) \; x(\tfrac k {V_\eps})\,  
\end{align*}
for $ \tfrac \ell {V_\eps} \in \Ueps$, and 0 else.
We therefore obtained a fully discrete model that is amenable to computer simulation,
\begin{equation}
\label{eq:y-discretized}
y_\ell \triangleq  \sum_{k\in[0,U V_\eps]}  \, h_{\ell , \ell-k}\,  x_k\end{equation}
for $\ell \in V_\eps\, \Ueps$.\footnote{Note that the values $h_{\ell , k}$ cannot be chosen arbitrarily.  In fact, they are determined in full by the values $h_{K\ell , k}$.}

To conclude, first note that this finite-dimensional model corresponds to the operator with time-varying impulse response 
\begin{align*}
\widetilde h(\tau,\nu) = \sum_{\tfrac m B \in \Ueps} \sum_{\tfrac n {V_\eps}\in[0,U]} \widetilde h (\tfrac m  {B}  ,\tfrac n {V_\eps}) \sinc (B t - m)\, \sinc (V_\eps \tau - n) 
\end{align*}
and spreading function
\begin{align*}
\widetilde \eta(\tau,\nu)= \sum_{\tfrac m B \in \Ueps} \sum_{\tfrac n {V_\eps}\in[0,U]} \widetilde h (\tfrac m  {B}  ,\tfrac n {V_\eps}) e^{-2\pi i \nu m/B}\chi(\tfrac \nu {B} )\, \sinc (V_\eps \tau - n). 
\end{align*}
Both are fully described by the (independent) samples 
\begin{align}
\eta_{k,\ell}             & = \widetilde \eta(k\Delta\tau, \ell \Delta \nu)
\end{align}
where $\Delta\tau= (U+2r(\eps))^{-1}$, $\Delta\nu= V_\eps^{-1}= (V+2r(\eps))^{-1}$, and integers $k$ and $\ell$ satisfying $k\Delta\tau \in [0,U]$, $\ell\Delta\nu \in [-\frac B 2, \frac B 2]$, respectively. 

Since $\widetilde \eta$ is obtained by convolving $\eta(\tau,\nu)$ with $V_\eps \sinc(V_\eps \tau)$ in $\t$, the relevant samples of $\eta(\tau,\nu)$ in $\tau$ are not restricted to $[0,T]$, but they are expected to decay away from $[0,T]$.  
As is customary, we shall consider operators (or sums thereof) with $\eta_{k,\ell} \neq 0$ only for $k\Delta\tau \in [0,T]$ and $\ell\Delta\nu \in [-\frac B 2, \frac B 2]$.

Let us emphasize that the grid widths $\Delta\tau$ and $\Delta\nu$ correspond to the size the area $[0,T]\times [-\tfrac B 2 , \tfrac B 2]$ that we are modeling the operator on, as well as on the allowed modeling error $\eps$.

Note that as
\begin{align*}
\widetilde H x(t) &= \int  \widetilde \sigma(t,f)\, \widehat x(f)\,e^{2\pi i t f} \d f \\
& = \int  \widetilde \sigma(t,f)\, \big(\widehat x(f)\chi(\tfrac{f}{V_\eps})\big)\,e^{2\pi i t f}\d f,
\end{align*}
we have $\widetilde Hx(t)=\widetilde H  (V_\eps\sinc(V_\eps \, \cdot)\ast x)(t)$ for all $x$. 
In particular, we can replace 
$\Shah(t) =   \sum_{k\in\Z}  c_{k  }  \delta_{kT} (t) $ by the smooth  identifier
\[
\Shah'(t) = V_\eps \sum_{k\in\Z}  c_{k  } \, \sinc{V_\eps t -  {V_\eps k T}}.
\]
Since we are only interested to study the operator's action for functions on $[0,U]$ in time, we have the added benefit of choosing the realistic
approximately time-gated identifier
\[
\widetilde \Shah(t) = V_\eps \sum_{\tfrac k {V_\eps}\in[0,U]}  c_{k} \, \sinc{V_\eps (t -  k T)}.
\]

In order to efficiently discretize the identifier $x(t)=\widetilde \Shah(t)$ by sampling on $(V_\eps)^{-1} \Z$, we have to postulate that $T \in {V_\eps}^{-1} \Z$ which is equivalent to $T V_\eps=\widetilde K \in\Z$.
Indeed, under this assumption, the samples of $\widetilde \Shah(t)$ form the
sequence $\{d_\ell\}$, which is nonzero only if $\frac \ell {V_\eps} = kT$ for some $\tfrac k {V_\eps}\in[0,U]$, which necessitates $\ell \in \widetilde
K\Z$.\footnote{Observe that $\tfrac {V_\eps} B=K$ and $TV_\eps=\widetilde K \in\Z$ imply that $TB=\tfrac {\widetilde K} K \in\mathbb Q$.}

In the stochastic case, the parameters determining the spreading function are random variables.
By the WSSUS assumption, its samples must be uncorrelated in both variables, giving rise to the discrete scattering function $C_{k,\ell}$ defined via
\begin{equation}\label{eq:discrete.WSSUS}
\EU*{\steta_{k,\ell}\: \conj{\steta_{k',\ell'} }} \coloneqq C_{k,\ell} \: \delta[k-k'] \: \delta[\ell-\ell'].
\end{equation}

We can now attempt to recover the discrete scattering function $C_{k,\ell}$ from the autocorrelation $\EU*{\y_\ell\, \conj{\y_{\ell'}}}$ of the finite
discrete echo $\y_\ell$ to the (finite) discrete sinc train $\widetilde\Shah$,
\begin{equation}
\label{eq:discrete.output}
\y_\ell \triangleq  \sum_{k\in[0,U V_\eps]}  \h_{\ell , \ell-k}\, d_k , \quad \ell \in V_\eps \, \Ueps,
\end{equation}
where the random variables $\h_{\ell,k}$ are determined by $\steta_{k,\ell}$ as described above.

\subsection{Test channels}\label{ssec:test.channels}
We assume $\H$ to be a time-varying target with a random spreading function $ \steta(\t,\nu)$ supported within a bounding rectangle $[0,\Tmax] \times
[0, \Bmax]$ with $\Tmax = \SI{1}{\s}$ and $\Bmax = \SI{3}{\Hz}$ of area three.  In accordance with the arguments in the previous subsection, we set
the scattering function $C(\t,\nu)$ (and hence, $\steta(\t, \nu)$) to zero outside of the bounding rectangle where necessary, and discretize it in
accordance with \eqref{eq:y-discretized}, letting $\steta_{k,\ell}$ to represent the target completely. 

\begin{figure}[ht]
\centering
\subfloat[$C_2(\t,\nu)$]{
\includegraphics[width=0.5\columnwidth]{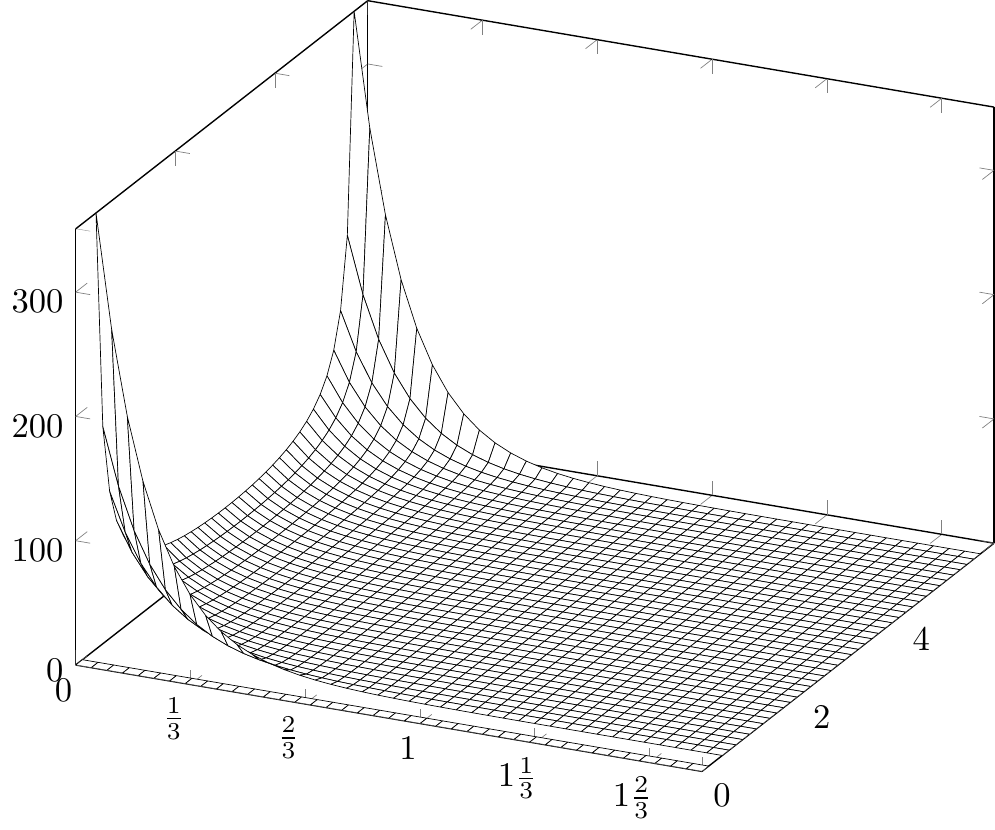}
% \includestandalone{tikz-C2-birdseye}
\label{fig:C2-birdseye}}
\\
\subfloat[$C_3(\t,\nu)$]{
\includegraphics[width=0.5\columnwidth]{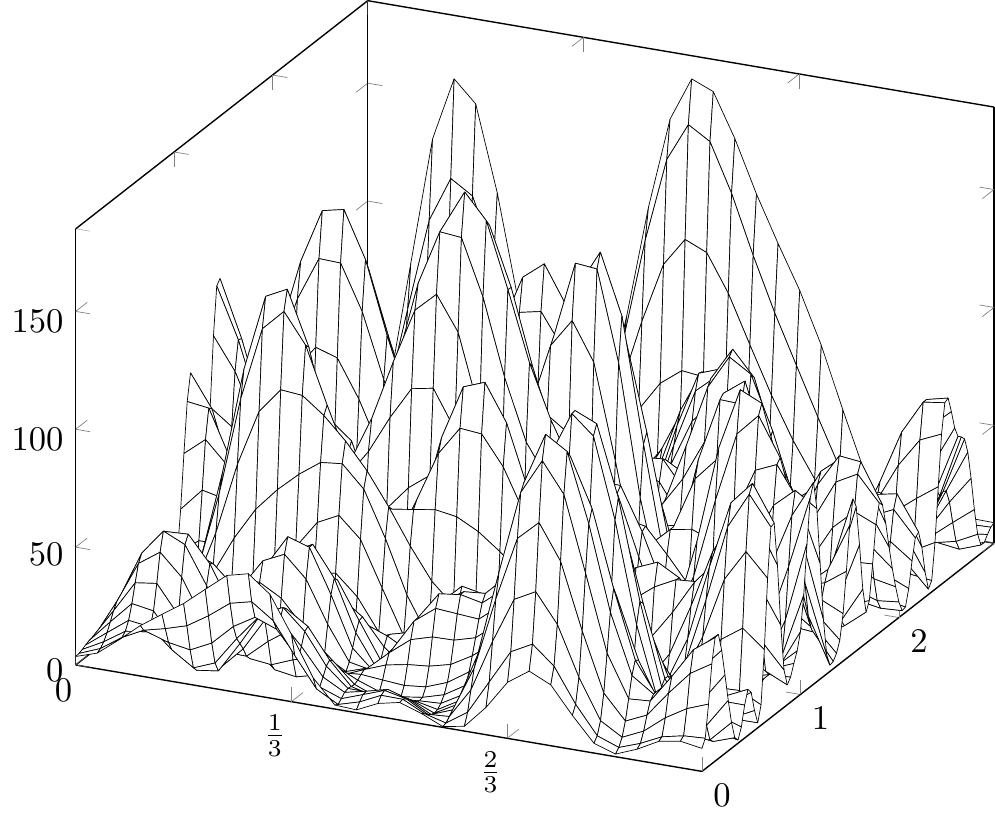}
% \includestandalone{tikz-C3-birdseye}
\label{fig:C3-birdseye}}
\caption{Birdseye view of scattering function in exponential-Jakes and a trigonometric polynomial with random coefficients models}
\end{figure}
%
%%%%%%%%%%%%%%%%%%%%%%%%%%%%%%%%%%%%%%%%%%%%%%%%%%%%%%%%%%%%
\begin{figure}[ht]%
\centering
\subfloat[Birdseye view of $C_1(\t,\nu)$]{
\includegraphics[width=0.5\columnwidth]{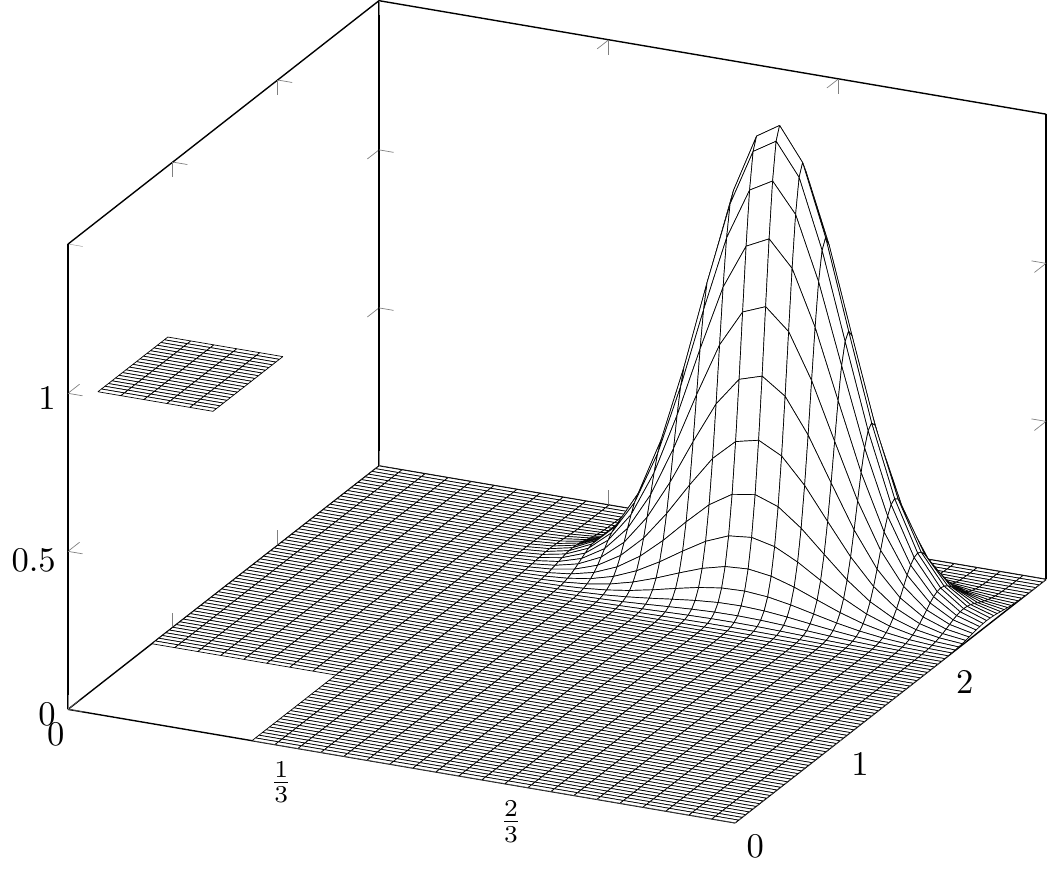}
% \includestandalone{tikz-C1-birdseye}
\label{fig:C1-birdseye}}
% \hfill%
\\
\subfloat[Heatmap view of $C_1(\t,\nu)$, with rectification of size $L=3$, with three active boxes selected]{
\includegraphics[width=0.5\columnwidth]{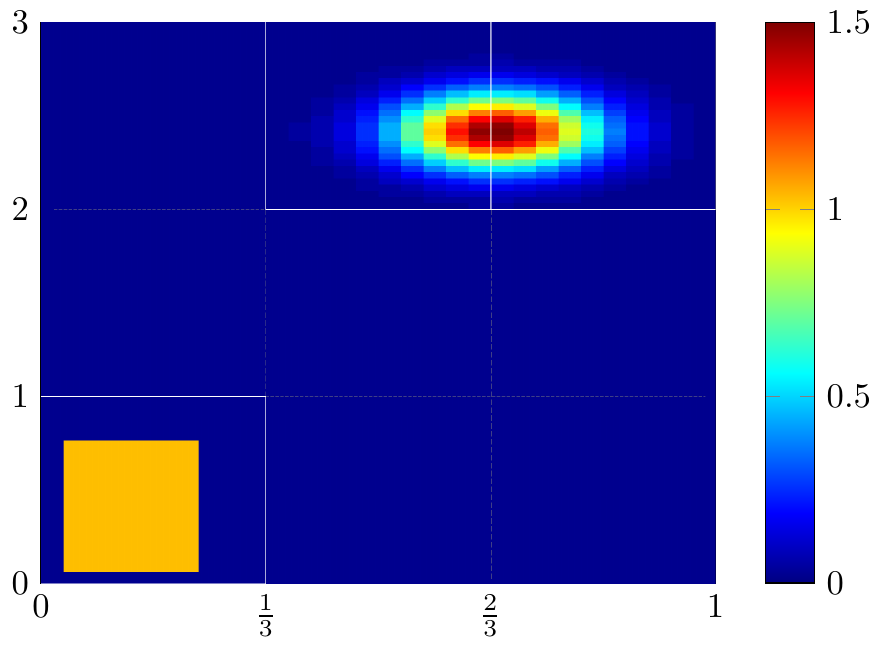}
% \includestandalone{tikz-C1-heatmap}
\label{fig:C1-heatmap}
}%
\caption{Model scattering function $C_1(\t,\nu)$.}
\end{figure}

For our simulation, we test the following models for the behavior of the scattering function of $\H$.
\begin{enumerate}
\item When $\H$ exhibits an approximate two-path propagation property, its scattering function may be modeled, for example, a sum of a Gaussian and a box function, shown in 3D on \autoref{fig:C1-birdseye} and as a heatmap on \autoref{fig:C1-heatmap},
\begin{equation}\label{eq:C.1}
C_1(\t,\nu) \triangleq \chi_{\rect}(\t, \nu) + \mathcal{N}(\mu_{\t}, \mu_{\nu}; \sigma_{\t}, \sigma_{\nu}),
\end{equation}
where the box is at $\rect = [0.12, 0.4]\times[0.1, 0.35]$ and the Gaussian has parameters $\mu_{\t} = 0.66, \mu_{\nu} = 2.4, \sigma_{\t} = 0.07, \sigma_{\nu} = 0.13$.
\item $\H$ follows a Jakes-exponential model as in \cite{Boashash},
\begin{equation}\label{eq:C.2}
C_2(\t,\nu) \triangleq \tfrac{1}{\rho^2} Q(\nu)\,  P(\t)
\end{equation}
with delay power profile
\[
P(\t) = \tfrac{\rho^2}{\tau_0} \, e^{-\t/{\tau_0}} \, \chi(\t/\Tmax)
\]
and Doppler power profile
\[
Q(\nu) = \tfrac{\rho^2}{\pi \sqrt{\Bmax^2 - \nu^2}} \, \chi(\nu/\Bmax).
\]
A scattering function with parameters $\rho=14$ and $\tau_0 = 0.3$, shifted to the first quadrant, is shown on \autoref{fig:C2-birdseye}.

\item $\H = \H''$ is a target with the spreading function given exactly by 
\begin{equation}
\label{eq:C.3}
C_3(\t,\nu) = \absq[\Big]{\sum_{m=-V/2}^{V/2} \sum_{n=0}^{U} c_{mn} \epi{(m \frac{\t}{\Tmax}  - n\frac{\nu}{\Bmax})}}
\end{equation}
with coefficients $\c_{mn}$ chosen at random, i.i.d. Gaussian.
A particular instance used is shown on \autoref{fig:C3-birdseye}, with $U = 30$, $V=5$.
\end{enumerate}

In all test cases, the area of the support cannot be substantially reduced by translating or adjusting the margins for tighter fit.

In every case, we are going to use the second order statistics of the discretized output $\{\y_{\ell}\}_{\ell \in V_\eps \Ueps}$ defined in \eqref{eq:discrete.output}
above to reconstruct the discrete scattering function $C_{k,\ell}$, which characterizes the second-order statistics of a WSSUS channel $\H''$.

For each of the cases $s = 1, 2,3$, we discretize the continuous-time scattering functions $C_s(\t,\nu)$ on a mesh $\Delta$, that is,
\begin{equation}\label{eq:scattering.grid}
(C_s)_{k,\ell} \coloneqq C_s(k \Delta \t, \ell \Delta \nu).
\end{equation}

We fix a sounding signal $x(t)$ to be a periodic impulse train as in \eqref{eq:tf-localized-signal}, with complex weights $c$ drawn either randomly
and uniformly from the unit circle prior to any sounding or set $c$ to be the fiducial vector of the appropriate dimension.  
For illustration, on \autoref{fig:x} we show two periods of $\x(t)$ ($t \in [0, 2]$), together with an instance of the output $\y(t) = H\, \x(t)$.
\begin{figure}[H]
\centering
% \includestandalone{tikz-comb}
\includegraphics{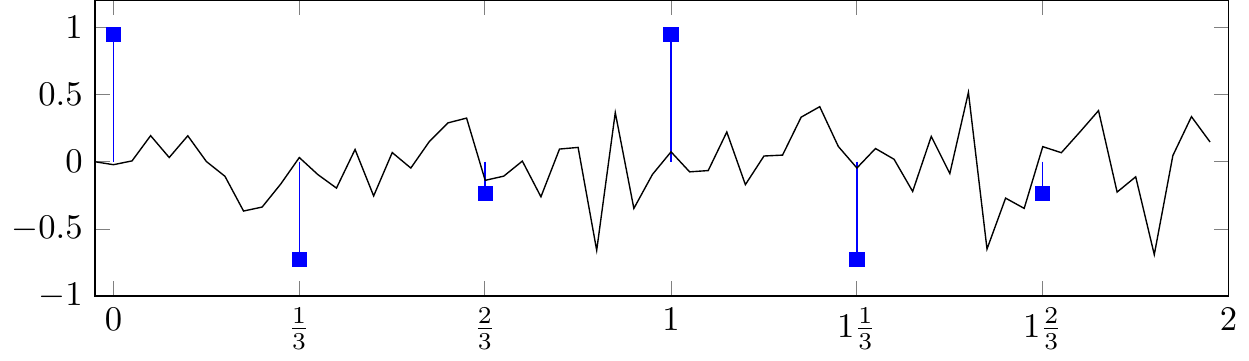}
\caption{Real parts of the input (the weighted impulse train $\x(t)$) and output signals $\y(t) = \H x(t)$ \label{fig:x}}
\end{figure}

We draw instances of the spreading function at random,
\[
\steta^{(j)}_{s; k,\ell} \sim \CNormal{0, (C_s)_{k,\ell}}, \quad  j=1,\dotsc,J,
\]
pointwise independent complex Gaussian, with mean zero and variances $C_{k, \ell}$ according to \eqref{eq:scattering.grid}; each instance representing a state of the channel.

For each sounding, we compute the output vector $y^{(j)}(t) = H^{(j)}\, \x(t)$ according to \eqref{eq:discrete.output} and subsequently the
discretized continuous-time Zak transform $\Zak y^{(j)}(\t,\nu)$ discretized on a mesh $\Delta$.\footnote{Notably, this is not the same as the discrete Zak transform.}
From the average \acorr of $\Zak y^{(j)}(\t,\nu)$ obtained by \eqref{eq:Dhat}, we estimate $\widehat{C}_s$ using the discretized version of \eqref{eq:C.Rz}.
For the spreading function $C_1(\t,\nu)$, we may choose the three active boxes highlighted in white, namely, $(1,1), (2,3), (3,3)$ as on \autoref{fig:C1-heatmap}, or attempt to use all nine boxes, as shown on \autoref{fig:C1-recon9-100}.
As discussed in \autoref{rem:Kfull-or-not-Kfull}, restricting attention to fewer boxes---even though our algorithm allows full recovery---is beneficial when the locations of the boxes are known, and the error induced by ignoring the almost empty boxes is negligible.
The effect of selecting three boxes rather than nine is shown on \autoref{fig:C1}.

\begin{figure}[H]
\centering
\includegraphics[width=0.5\columnwidth]{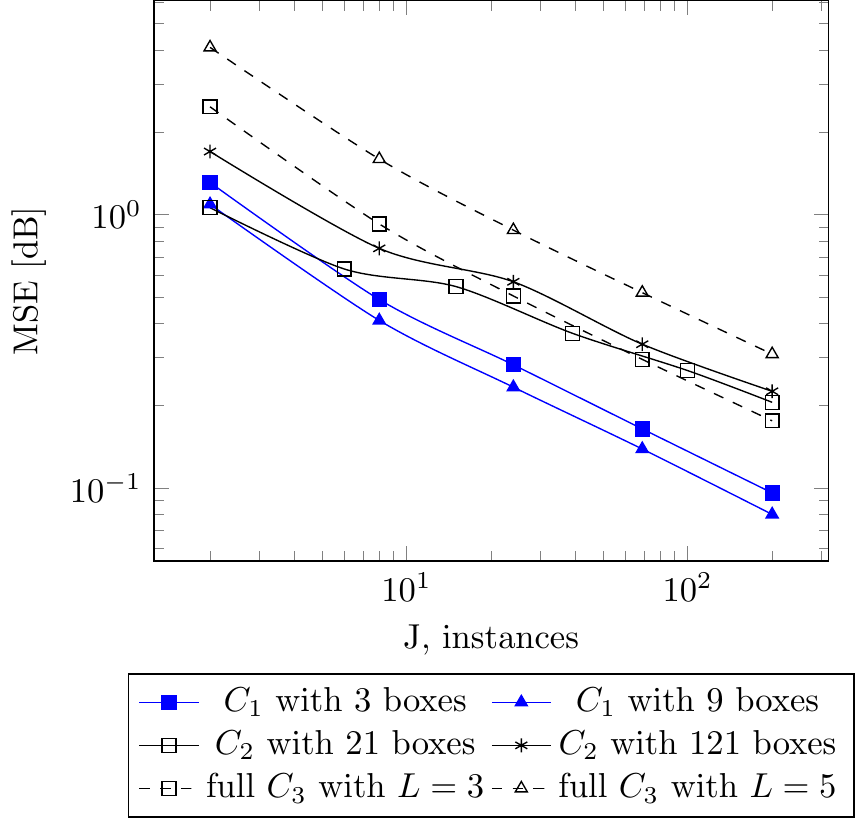}
% \includestandalone{tikz-relative-mse}
\caption{Relative MSE $\frac{\norm{\hat{C}-C}_2}{\norm{C}_2}$ as a function of $J$. \label{fig:MSE.two}}
\end{figure}

The mean-square error of reconstruction for the scattering functions given by $C_2(\t,\nu)$ and $C_3(\t,\nu)$ are also given on \autoref{fig:MSE.two}.
In our tests, the estimator behaved as expected, with error decaying linearly in the number of soundings $J$, as predicted by \autoref{eq:var.bound}.
For simulated cases $C_1(\t,\nu), C_2(\t,\nu), C_3(\t,\nu)$, the normalized mean square error of the estimator is shown on \autoref{fig:MSE.two}.

\section{Conclusion}\label{sec:conclusion}
We show that it is possible to recover the scattering function of a target using the second order statistics of the returned echoes from repeated sounding by a custom weighted delta train, as long as its scattering function has bounded support.
This includes overspread targets, under both the classic criterion that the area of the bounding box is less than one, and the sharper condition that the area of the support itself  is less than one.

We exhibit the universal deterministic choice for a family of weights that provide identifying pulse trains for any support size and level of resolution, given by so called fiducial vectors.
This provides a nice link to the theory of equiangular tight frames and gives fiducial vectors another useful application.
The excellent numerical properties of fiducial vectors allow a more flexible indiscriminate sounding regime, removing the need for a probabilistic selection of the weight vector in the deterministic case, a nuisance even in cases when a probabilistic choice is known to succeed with overwhelming probability.

We give an explicit recipe of a simple unbiased estimator, with predictable behavior of its variance given repeated soundings under the assumption of pointwise jointly proper Gaussianity of the scatterers.
We show that the incurred error decays inversely proportionate to the number of soundings.
We give extensive simulation results, confirming our theoretical claims empirically.
We explore some numerical aspects of it, including the effect of geometry-conscious identification, using probabilistic vs.\ fiducial weight vectors, and the contamination from neglected approximately zero areas.

%%%%%%%%%%%%%%%%%%%%%%%%%%%%%%%%%%%%%%%%%%%%%%%%%%%%%%%%%%%% 
\begin{figure*}[t]%
\centering
\subfloat[$\abs{C_1^{(1)} - C_1}$: error after 1 sounding, using 3 boxes]{  
\includegraphics[width=0.45\columnwidth]{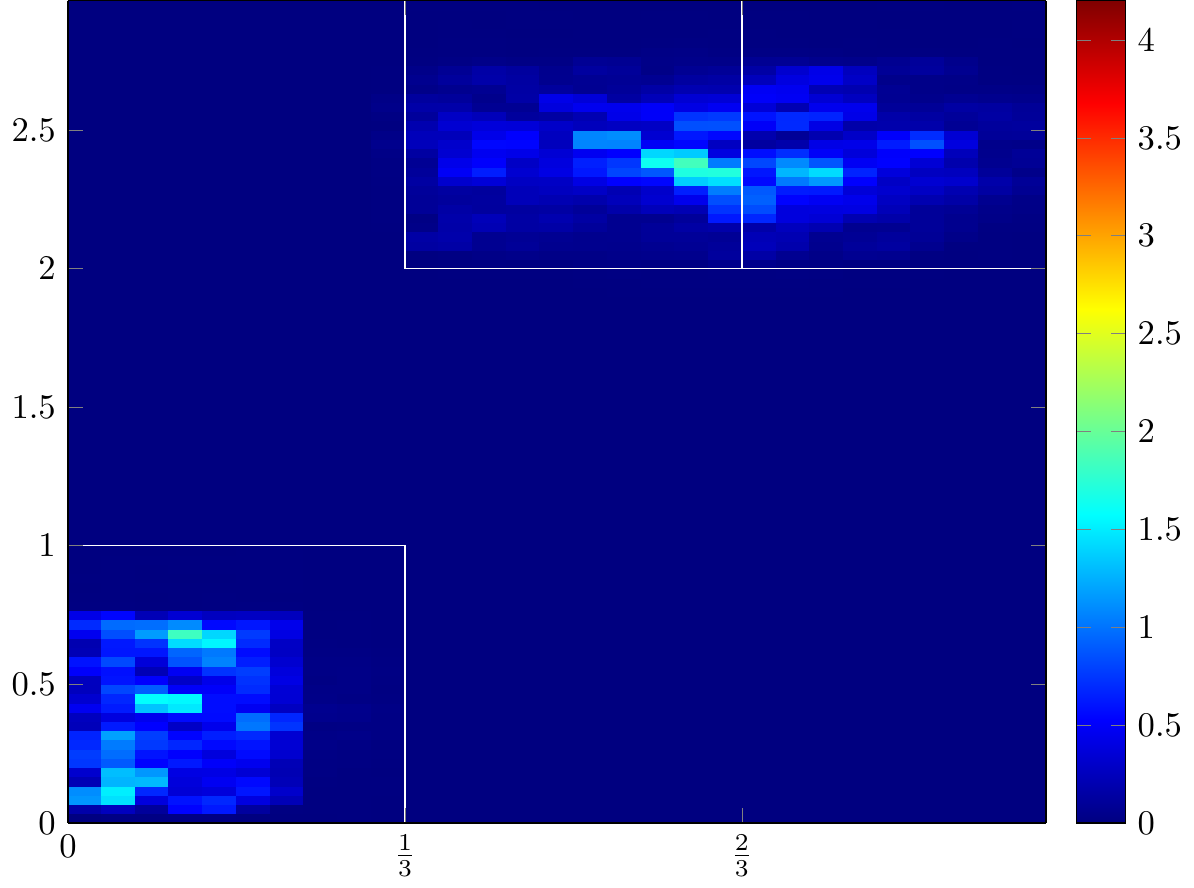}
% \includestandalone{tikz-C1-diff-heatmap}
\label{fig:C1-recon1}
}%
\hfill
\subfloat[$\abs{C_1^{(200)} - C_1}$: error after 200 sounding, using 3 boxes]{
% \includestandalone{tikz-C1-200-heatmap}
\includegraphics[width=0.45\columnwidth]{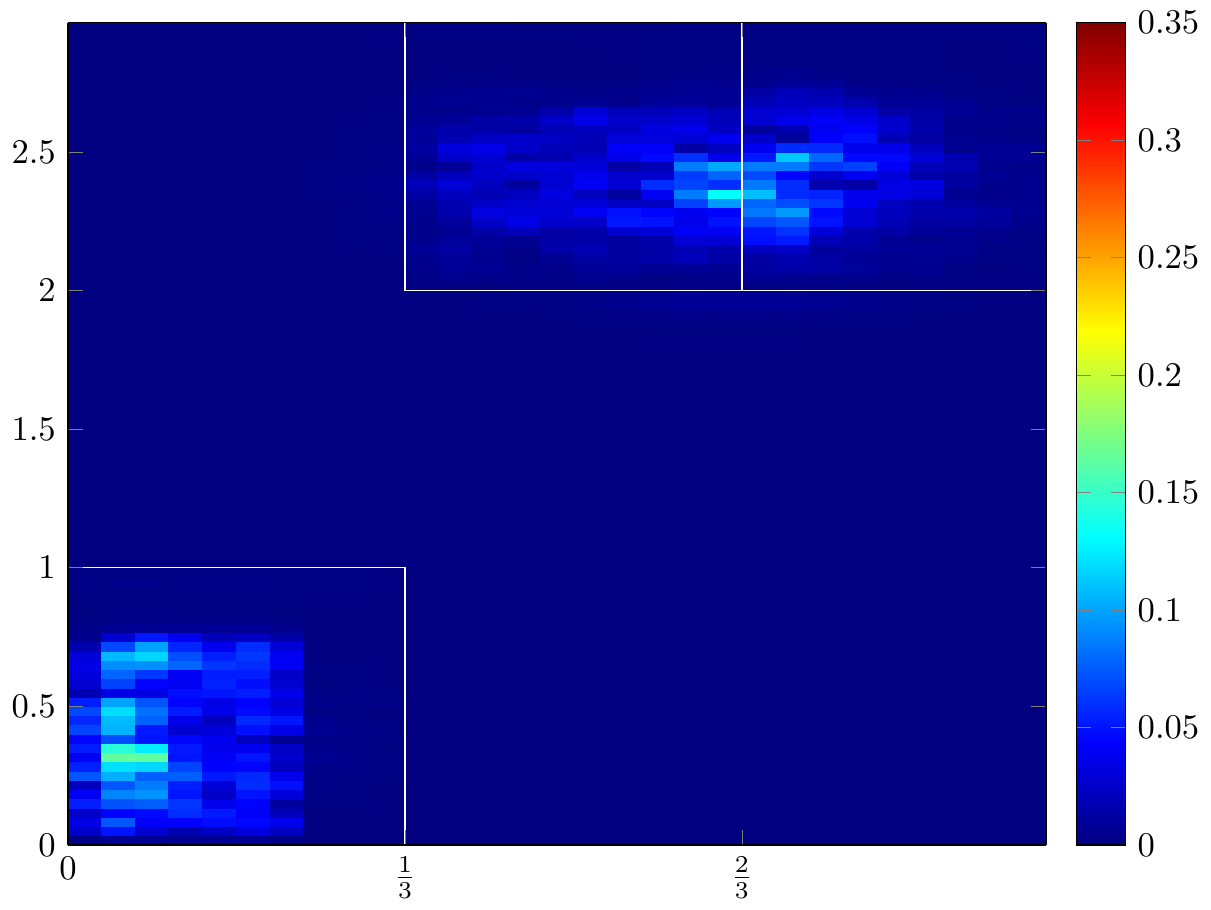}
\label{fig:C1-recon200}
}%
% \end{figure*}
% \begin{figure*}[t]
% \centering
\\
\subfloat[$\abs{C_1^{(1)} - C_1}$: error after 1 sounding, using 9 boxes]{
\includegraphics[width=0.45\columnwidth]{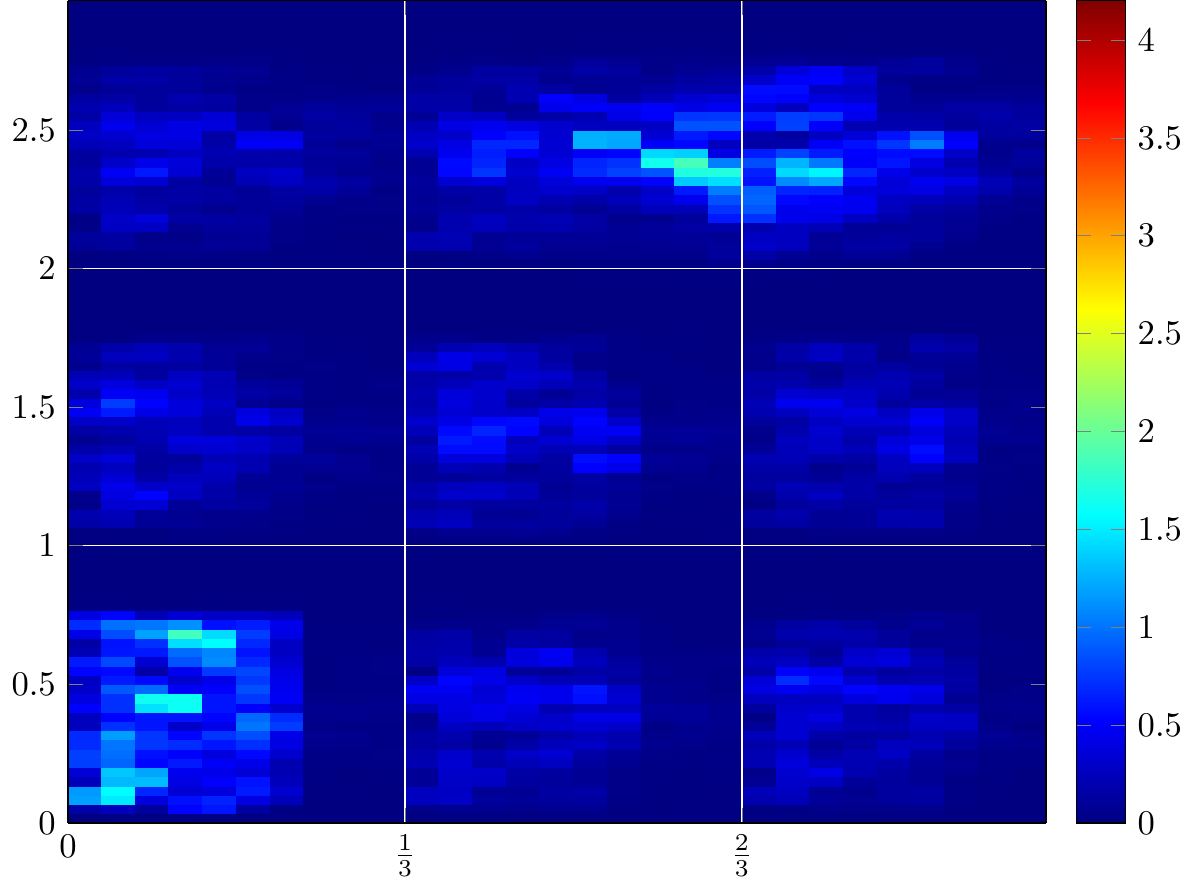}
% \includestandalone{tikz-C1-diff-1instance-heatmap}
\label{fig:C1-recon9-1}
}
\hfill  
\subfloat[$\abs{C_1^{(200)} - C_1}$: error after 200 sounding, using 9 boxes]{
\includegraphics[width=0.45\columnwidth]{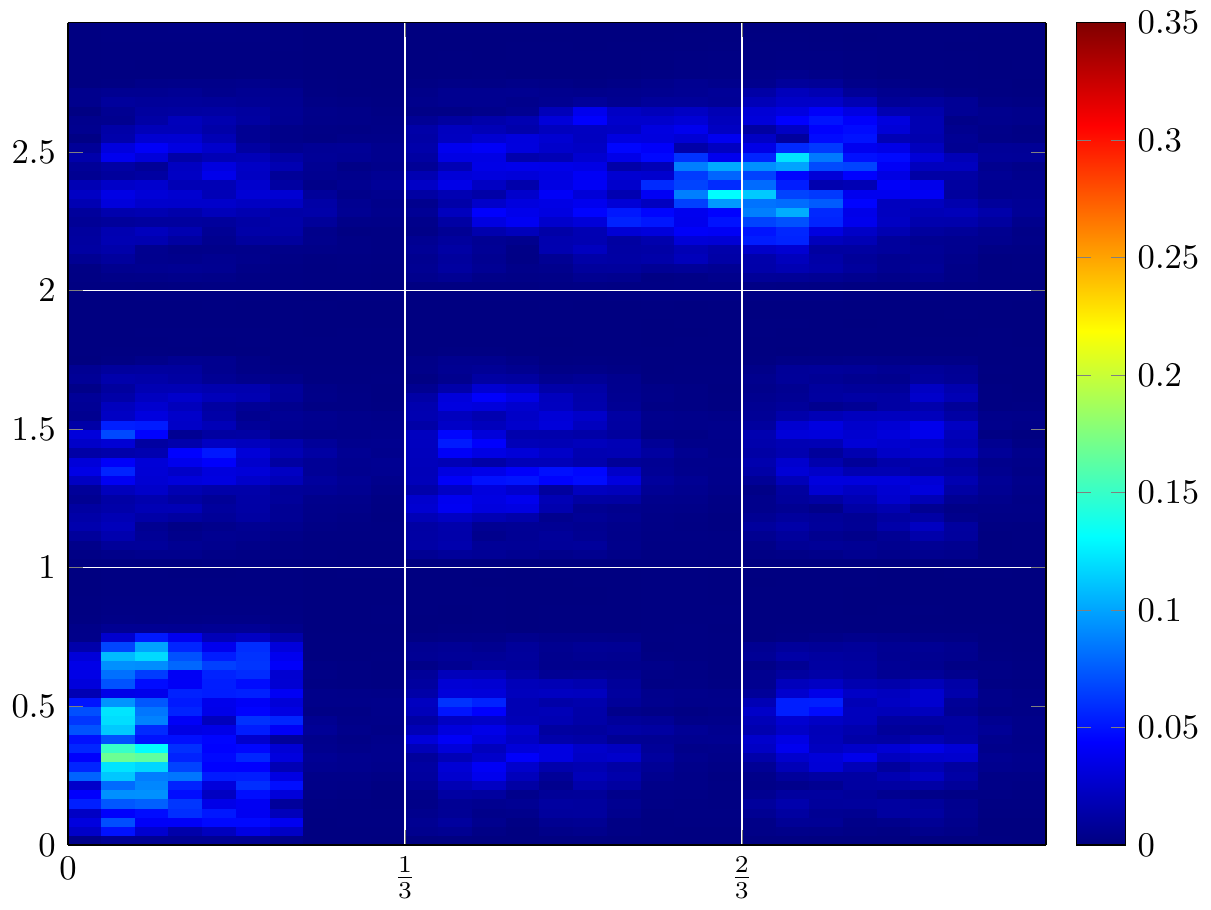}
% \includestandalone{tikz-C1-diff-200instances}
\label{fig:C1-recon9-100}
}
\caption{$C_1(\t,\nu)$, reconstruction under different scenarios} 
\label{fig:C1}%
\end{figure*}

\bibliographystyle{hieeetr}
% \bibliography{ScatteringFunction}

\end{document}